\documentclass[copyright,creativecommons]{eptcs} 
\providecommand{\event}{ICE 23} 
\usepackage{underscore}           

\usepackage{mathtools,mathpartir,amsmath,amssymb}
\usepackage{hyperref}
\usepackage{keyval}
\usepackage{prooftree}
\usepackage[all]{xy}
\usepackage{dsfont}
\usepackage{tikz}                  
\usepackage{enumerate}                  
\usetikzlibrary{shadows,arrows,shapes,automata,positioning,decorations.markings,fit}
\usepgflibrary{arrows.spaced}
\usepackage{wrapfig}
\usepackage{breakurl}
\usepackage[misc,geometry]{ifsym}
\usepackage[algoruled,vlined]{algorithm2e}
\SetKwProg{Fn}{Function}{}{}

\usepackage[nomargin,inline,index,%
]{fixme}

\usepackage[shortlabels]{enumitem}
\usepackage{makeidx}
\usepackage{color}

\usepackage{mathpartir,amsmath,amssymb
}

\newtheorem{definition}{Definition}[section]
\newtheorem{lemma}[definition]{Lemma}

\newtheorem{theorem}[definition]{Theorem}

\newtheorem{remark}[definition]{Remark}
\newtheorem{example}[definition]{Example}

\newtheorem{fact}[definition]{Fact}
\newenvironment{proof}{{\em Proof.}}{
}

\def \bfr {\ \begin{color}{blue}} 
\def \efr {\end{color}\ }
\def \bmr {\begin{color}{red}} 
\def \emr {\end{color}}

\newcommand{\msg}[1]{\textsc{\footnotesize #1}}

\newcommand{\pair}[2]{(#2,#1)}
\newcommand{\set}[1]{\{#1\}}
\newcommand{\rn}[1]{\textsc{\footnotesize [{#1}]}}
\newcommand{\coDef}{::=^{coind}}

\newcommand{\qed}{\hfill $\Box$}
 \newcommand{\Nset}
                          {\mathcal{N}}

\newcommand{\G}{{\sf G}}
\newcommand{\Gc}[6]{\ensuremath{#1\to#2:\set{#3_#5. #4_#5}_{#5\in#6} }}

\newcommand{\End}{\ensuremath{\mathtt{End}}}

\newcommand{\inact}{\ensuremath{\mathbf{0}}}
\newcommand{\PP}{P}
\newcommand{\Q}{Q}
\newcommand{\R}{R}

\newcommand{\U}{U}
\newcommand{\Pin}[5]{\ensuremath{#1?\set{#2_#4. #3_#4}_{#4\in#5} }}
\newcommand{\Pou}[5]{\ensuremath{#1!\set{#2_#4. #3_#4}_{#4\in#5} }}

\newcommand{\Pb}[5]{\ensuremath{#1\dagger\set{#2_#4. #3_#4}_{#4\in#5} }}

\newcommand{\pa}{{\sf a}}
\newcommand{\pb}{{\sf b}}
\newcommand{\pc}{{\sf c}}
\newcommand{\pp}{{\sf p}}
\newcommand{\q}{{\sf q}}

\newcommand{\pr}{{\sf r}}
\newcommand{\ps}{{\sf s}}

\newcommand{\pu}{{\sf u}}

\newcommand{\PB}{B}

\newcommand{\PS}{S}
\newcommand{\PU}{U}
\newcommand{\PQ}{Q}
\newcommand{\PR}{R}

\newcommand{\participants}{\mathcal{P}}

\newcommand{\pP}[2]{\ensuremath{#1[#2]}}
 \newcommand{\parN}{\mathrel{\|}}

 \newcommand{\tyng}[3]{#1\vdash_{#3} #2}
 
 \newcommand{\NamedCoRule}[5][]{\IInfer[#1]{#2}{ #3 }{#4}{#5}} 
\newcommand {\IInfer} [5] [] {
  \inferrule*[%
    fraction={===}, 
    left={\textsc{#2}},%
    right={$\begin{array}{l} #5 \end{array}$}, 
    #1
  ]%
   {#3}{#4}}
  \newcommand{\NamedRule}[5][]{\Infer[#1]{#2}{ #3 }{#4}{#5}} 
\newcommand {\Infer} [5] [] {
  \inferrule*[%
    left={\textsc{#2}},%
    right={$\begin{array}{l} #5 \end{array}$}, 
    #1
  ]%
   {#3}{#4}}

   \newcommand{\Nt}{{\mathbb{M}}}

   \newcommand{\X}{X}

   \newcommand{\plays}[1]{\ensuremath{{\sf prt}(#1)}}
   \newcommand{\Y}{Y}
   
   \newcommand{\Set}[1]{\{ #1 \}}

    \newcommand{\DD}{\mathcal D}
     
     \newcommand{\depth}[2]{ \textit{depth} (#1, #2)}
     \newcommand{\mypath}{\sigma}

  \newcommand{\Nat}{\mathbb{N}}
  \newcommand{\paths}[1]{\textit{paths} (#1)}
  \def\finex{{\unskip\nobreak\hfil
\penalty50\hskip1em\null\nobreak\hfil{\Large $\diamond$}
\parfillskip=0pt\finalhyphendemerits=0\endgraf}}

\newcommand{\stackred}[1]{\xrightarrow{#1}}
\newcommand{\concat}[2]{\ensuremath{#1\,{\cdot}\,#2}}
\newcommand{\cardin}[1]{ | #1 |}
   \newcommand{\ee}{\epsilon}
   
\renewcommand{\Nset}{\mathcal{N}}
\newcommand{\itr}[1]{$\iota$-#1}
\newcommand{\tynI}[3]{#1\vdash_{#3}^{\iota} #2} 
\newcommand{\tynIP}[3]{#1\vdash_{#3}^{\iota} #2}
\newcommand{\itrp}[1]{$\iota$-#1}

\newcommand{\tp}[2]{{\sf tp}(#1,#2)}

\newcommand{\Gpat}{\mathbb{G}} 
\newcommand{\Ppat}{\mathbb{P}}
\newcommand{\vars}[1]{\mathsf{vars}(#1)}
\newcommand{\sbtplus}{+} 
\newcommand{\dom}[1]{\mathsf{dom}(#1)}
\newcommand{\sbtord}{\preceq}
\newcommand{\agteq}[2]{#1 \eqcirc #2}
\newcommand{\eqsys}{\mathcal{E}} 
\newcommand{\eqsysP}{E} 
\newcommand{\gsol}[2][]{\mathsf{sol}_{#1}(#2)}
\newcommand{\tyalg}[4]{#1 \vdash #2 \Rightarrow (#3,#4)} 
\newcommand{\Goals}{\mathcal{S}}
\newcommand{\infn}[1]{A-#1}

\newcommand{\la}{\M}
\newcommand{\M}{\lambda}

\newcommand{\ipair}[3]{(#1,#2,#3)}

\newcommand{\sw}{Weak} 
\newcommand{\x}{x}  
\newcommand{\y}{y} 
\newcommand{\pcon}[5]{(#1\cup#2)\setminus\set{#3,#4}\doteq#5}
\newcommand{\pcone}[5]{(#1\cup#2)\setminus\set{#3,#4}=#5}
\newcommand{\Pcon}{\mathcal C}
\newcommand{\triple}[3]{(#1,#2,#3)}
\newcommand{\Lc}{Lock}

\newcommand{\Cline}[1] {\vspace{1.4mm}\centerline{$ #1 $}\vspace{1.5mm}}

\begin{document}

\title{Partially Typed Multiparty Sessions}
\author{Franco Barbanera\thanks{
Partially supported by 
Project "National Center for "HPC, Big Data e Quantum Computing",  Programma M4C2 – dalla ricerca all’impresa – Investimento 1.3: Creazione di “Partenariati estesi alle università, ai centri di ricerca, alle aziende per il finanziamento di progetti di ricerca di base” – Next Generation EU; and by the 
Piano Triennale Ricerca Pia.Ce.Ri UniCT.
}
\institute{
Dipartimento di Matematica e Informatica,
Universit\`a di  Catania, Catania, Italy}
\email{barba@dmi.unict.it}
\and
Mariangiola Dezani-Ciancaglini
\institute{
Dipartimento di Informatica,
Universit\`a di Torino, Torino, Italy}
\email{dezani@di.unito.it}
}

\def\titlerunning{Partially Typed Multiparty Sessions} 

\def\authorrunning{ Barbanera \& Dezani-Ciancaglini
}

\maketitle

\begin{abstract}
 A multiparty session 
 formalises a set 
of concurrent communicating participants.
We propose a type system for 
multiparty sessions where some communications between participants can be  
ignored. This allows us to type some sessions with global types
representing interesting protocols, which have no type in the standard type systems. 
Our type system enjoys Subject Reduction, Session Fidelity and  ``partial'' Lock-freedom.
The last property  ensures  
the absence of 
\emph{locks} 
for participants with non ignored 
communications. A sound and complete type inference algorithm is also discussed. 
\end{abstract}

\section{Introduction}\label{int}

The key issue in multiparty distributed systems is the composition of independent entities such that
a sensible behaviour of the whole emerges from  those  
of the components, while avoiding type errors 
of exchanged messages and ensuring good  communication  properties like  Lock-freedom. 
MultiParty Session Types (MPST), introduced
in~\cite{HYC08,Honda2016}, are a class of choreographic formalisms for the description and analysis of such systems. 
Choreographic formalism are characterised by the coexistence of two
distinct but related views of distributed systems: the \emph{global}
and the \emph{local} views. 
The former describes the behaviour of a system as a whole, whereas the local views specify the behaviour of the single components in ``isolation''.
Systems described 
by means of MPST formalisms are usually  ensured 
(i) their overall behaviour to adhere to a given communication protocol (represented as a global type) and
(ii) to enjoy particular communication properties like  Lock-freedom (the 
specific property we focus on in the present paper). 

In~\cite{BDL22}  a  MPST formalism was developed for systems using synchronous communications,
where global types can be assigned to multiparty sessions (parallel composition of named processes)
via a type system. Typability of a multiparty session $\Nt$ by a global type $\G$  ensures  
that $\Nt$ behaves as described by $\G$ and is lock-free.

The property of  Lock-freedom  ensures that no lock is ever reached in the evolution of a system.
A lock is a system's reachable configuration where a  participant,  
which is able to perform
an action, is forever prevented to do so in any possible continuation. In particular,
such a configuration is called a $\pp$-lock  in case the stuck participant be $\pp$.
Lock-freedom -- which entails  Deadlock-freedom -- 
could be  however too strong to be proved in some  settings,  
and actually useless  sometimes.
As a matter of fact,  for  
particular systems,
the presence of  $\pp$-locks for some participants would not be problematic and 
would not break their  specifications. 
Let us assume, for instance, to have a social medium where participants can ask for
upgrades of their communication level (i.e. the capability describing  which participants they can communicate with  and 
which sort of  messages 
 can be sent). The upgrades are granted by some particular participant
$\pu$ according to the particular policy of the social medium. 
In case $\pu$ be implemented so to reply to an unbounded number of requests,
it is immediate to realise that, in case all the participants get to the highest communication
level, we would be in presence of a $\pu$-lock since no more level upgrade will be
requested. This would not be a problem, since what we are interested in is the possibility
for all the participants to progress until no communication with $\pu$ is possible.
From that moment on the participants other than $\pu$ must be  ensured  
to progress, but not  $\pu$.
This sort of circumstance is typical in clients/servers scenarios.
Given a set of participants $\participants$, we dub a system to be 
{\em $\participants$-excluded lock-free} whenever  it is $\pp$-lock free for each participant $\pp$ not belonging to $\participants$.

In  this paper we present a MPST  type  system in the style of~\cite{BDL22}
where it is possible to derive judgments of the  new  shape\\  
\Cline{
\tyng{\G}{\Nt}{\participants}
}

 We say that our typing is {\em partial} since some communications between participants in $\participants$ do not appear in the global type. Our type system 
  ensures that $(a)$ the communications of the participants  in 
 $\Nt$ not belonging to  $\participants$ comply with the interaction scenario represented 
 by $\G$ and $(b)$ $\Nt$ is $\participants$-excluded lock-free.

{\em Contributions and structure of the paper.} In Section~\ref{cal} we 
 recall  
the calculus of  multiparty sessions  from ~\cite{BDL22},   together with  
the global types. 
  Also, we introduce the novel notion of $\participants$-excluded \Lc-freedom, that we 
clarify by means of an example.     Section~\ref{ts}  is 
devoted to the presentation of our ``partial'' type system, assigning global types to multiparty sessions  where some communications can be ignored.  
 Besides, we prove   
the relevant properties of  partially  typable sessions:  Subject Reduction, Session Fidelity and
 $\participants$-excluded   \Lc-freedom. 
In Section~\ref{ti} we discuss a sound and complete type inference algorithm for our partial type system.
A section summing up our results, 
discussing related works and possible directions for future work concludes  the paper.

\section{Multiparty Sessions and Global Types}\label{cal}

 In this section we recall the calculus of multiparty sessions  and the global types  defined in~\cite{BDL22}. This calculus is simpler 
than the original  MPST calculus~\cite{HYC08}  and many of the subsequent ones.
Lack of  explicit  channels -- even if preventing the representation of session interleaving and delegation -- 
enables us to focus on our main concerns and allows for 
a clear explanation of the type system we will introduce in the next section.

We use the following base sets and notation: \emph{messages}, ranged
over by $\lambda,\lambda',\dots$; \emph{session participants}, ranged over
by $\pp,\q,\pr, \ps, \pu,\ldots$; \emph{processes}, ranged over by $\PP,\Q,\PR,\PS,\U,\dots$;
\emph{multiparty sessions}, ranged over by $\Nt,\Nt',\dots$;
\emph{integers}, ranged over by $i, j,l,h,k,
\dots$;   {\em  integer  
sets}, ranged over by $I,J,L, H, K,\dots$. 

 \begin{definition}[Processes]\label{p} 
  {\em Processes} are defined by:
  
  \Cline{\PP\coDef\inact\ \mid\ \Pou\pp\lambda\PP i I\ \mid\ \Pin\pp\lambda\PP i I}
  
\noindent
where $I\neq\emptyset$ and  $\lambda_j\neq\lambda_h$ for  $j, h\in I$ and $j\neq h$.  
\end{definition}
The symbol $\coDef$ 
in Definition~\ref{p} and in later definitions indicates that the
productions are interpreted  \emph{coinductively}.
 That is,  processes are possibly infinite terms. 
However, we assume such processes to be \emph{regular},  i.e.,  with 
finitely many distinct sub-processes.  
  This is done also in~\cite{CastagnaGP09} and it allows us to adopt in proofs the coinduction style
advocated in~\cite{KozenS17} which, without any loss of formal rigour,
 promotes readability and conciseness. 

Processes implement the  communication behaviour of participants.
The output process
$ \Pou\pp\lambda\PP i I$  non-deterministically 
chooses one message $\lambda_i$ for some $i \in I$, and
sends it to the participant $\pp$, thereafter continuing as $\PP_i$. 
Symmetrically, the input process $ \Pin\pp\lambda\PP i I$ waits for one of
the messages $\lambda_i$ from the participant $\pp$, then continues as $\PP_i$
after receiving it.  When there is only one output we write $\pp!\lambda.\PP$ and similarly for one input.  We use $\inact$ to denote the terminated process.  We shall omit writing trailing $\inact$s in processes. We denote by $ \Pb\pp\lambda\PP i I$ either $ \Pou\pp\lambda\PP i I$ or $ \Pin\pp\lambda\PP i I$. 

In a full-fledged calculus, messages would carry values,  that 
we avoid for the sake of simplicity; hence
no selection operation over values is included in the  syntax.  

\begin{definition}[Multiparty sessions] 
{\em  Multiparty sessions}  are expressions of the
shape: 

\Cline{
\pP{\pp_1}{\PP_1} \parN \cdots \parN \pP{\pp_n}{\PP_n}} 

\noindent
where  $\pp_j \neq \pp_h $ for $1\leq j,h\leq n$ and $j\neq h$. 
 We use $\Nt$ to range over
multiparty sessions. 
\end{definition}

Multiparty sessions  (sessions, for short)  are parallel compositions of located processes of the form
$\pP{\pp}{\PP}$, each enclosed within a different
participant $\pp$. We assume the standard structural congruence  $\equiv$  on
multiparty sessions, stating that
parallel composition is associative and commutative and has neutral
elements $\pP\pp\inact$ for any $\pp$.
If $\PP\neq\inact$ we write $\pP{\pp}{\PP}\in\Nt$ as short for
$\Nt\equiv\pP{\pp}{\PP}\parN\Nt'$ for some $\Nt'$. This abbreviation
is justified by the associativity and commutativity of parallel composition. 

 The {\em set of  active  participants} (participants for short) of a session $\Nt$, notation $\plays\Nt$, is   as expected:

\Cline{\plays\Nt=\set{\pp\mid \pP{\pp}{\PP}\in\Nt}}

 \noindent It is easy to verify that the sets of participants of structurally  congruent  
  sessions coincide.   

To define the {\em  synchronous operational semantics} of sessions we use  an  LTS,
whose transitions are decorated by  labels denoting message exchanges. 

\begin{definition}[LTS for Multiparty Sessions]\label{slts}
The {\em   labelled transition   system   (LTS) 
for multiparty sessions}  
  is the closure under structural congruence of the reduction specified by the unique rule:
  
  \Cline{\begin{array}[c]{@{}c@{}}
      \NamedRule{\rn{Comm-T}}{  h\in I \subseteq J 
      }{
      \pP\pp  {\Pou\q\lambda\PP i I}
      \parN
      \pP \q {\Pin\pp\lambda\Q j J }\parN\Nt
      \stackred{\pp \lambda_h \q}
      \pP\pp{\PP_h}\parN\pP\q{\Q_h}\parN\Nt
      }{}
    \end{array}
 }
 
\end{definition}
\noindent
Rule \rn{Comm-T} makes 
 communications possible, by describing 
when a participant $\pp$ can send a message $\lambda_h$ to participant $\q$,
and what is the effect of such message exchange.
%
This rule is non-deterministic in the choice of messages.
The condition $I \subseteq J$
ensures that the sender can freely choose the message, since 
the receiver must offer all sender messages and possibly more.
This allows us to distinguish in the operational semantics between
internal and external choices.
  Note that this condition will 
  be ensured by the typing Rule \rn{Comm} (see Definition~\ref{def:type-system}). 
 
 Let $\Lambda$  range over   \emph{labels}, namely triples   
 of the form 
 $\pp\lambda\q$. We define {\em traces} as (possibly infinite) sequences of  labels 
 by:\\
\Cline{\mypath\coDef\ee\mid\concat\Lambda\mypath}

\noindent
 where $\ee$ is the empty sequence. 
We use $\cardin{\mypath}$ to denote the length of the trace $\mypath$, where $\cardin{\mypath} = \infty$ when $\mypath$ is an infinite trace.
 We define the participants of  labels  and traces:
 
 \Cline{\plays{\pp\lambda\q}=\set{\pp,\q}\qquad\plays\ee=\emptyset\qquad\plays{\concat\Lambda\mypath}=\plays\Lambda\cup\plays\mypath}
 
\noindent
 When $\mypath=\concat{\Lambda_1}{\concat\ldots{\Lambda_n}}$ ($n\geq 0)$
we write $\Nt\stackred{\mypath}\Nt'$ as short for
%
$\Nt\stackred{\Lambda_1}\Nt_1\cdots\stackred{\Lambda_n}\Nt_{n}  =  \Nt'$. 
 As usual we write $\Nt\rightarrow$  (resp. $\Nt\not\rightarrow$) when there exist (resp. no)  $\Lambda$ and $\Nt'$ such that 
$\Nt \stackred{\Lambda}\Nt'$.
 
  It is easy to verify that,  in a transition, only the two participants of its
  label are  involved,
   as formalised  below. 
 
 \begin{fact}\label{fact:arb}
If $\set{\pp,\q}\cap\set{\pr,\ps}=\emptyset$ and $\pP\pr\PR\parN\pP\ps\PS\parN\Nt  \stackred{\pp \lambda \q} \pP\pr\PR\parN\pP\ps\PS\parN\Nt'$, then 

\Cline{\pP\pr{\PR'}\parN\pP\ps{\PS'}\parN\Nt  \stackred{\pp \lambda \q} \pP\pr{\PR'}\parN\pP\ps{\PS'}\parN\Nt'}

\noindent for arbitrary $\PR',\PS'$.
\end{fact}

We  define now  the property of  $\participants$-excluded  \Lc-freedom,  
 a ``partial'' version of  the standard \Lc-freedom~\cite{Kobayashi02,Padovani14}. 
The latter consists in the  possible  eventual completion of pending communications of any participant
(this can be alternatively stated by saying that any participant is lock-free).
We are interested instead in the  progress of 
some specific participants only, 
 namely those we decide not to ``ignore''. 
  In the following, $\participants$ will range over sets of  ignored participants.

\begin{definition}[$\participants$-excluded \Lc-freedom]\label{d:lf}
   A  
  multiparty session $\Nt$ is a {\em $\participants$-excluded  lock-free} session if
	 $\Nt\stackred{\mypath}\Nt'$ and  $\pp\in\plays{\Nt'}\setminus\participants$  imply $\Nt'\stackred{\concat{\mypath'}\Lambda}\Nt''$
	 for some $\mypath'$ and $\Lambda$ such that $\pp\in\plays\Lambda$.
\end{definition}

 Clearly the standard notion of Lock-freedom coincides with $\emptyset$-excluded Lock-freedom. 

\smallskip

It is natural to extend also the usual notion of  Deadlock-freedom  to our setting.
\begin{definition}[$\participants$-excluded Deadlock-freedom]
  A  
  multiparty session $\Nt$ is a {\em $\participants$-excluded  deadlock-free} session if
	 $\Nt\stackred{\mypath}\Nt'\not\rightarrow$  implies  $\plays{\Nt'}\subseteq\participants$. 
\end{definition}

It is immediate to check that, as for standard  \Lc- and Deadlock-freedom,  the following 
hold.
\begin{fact}
$\participants$-excluded  \Lc-freedom  implies $\participants$-excluded  Deadlock-freedom. 
\end{fact} 
  The vice versa does not hold. For example if $\PP=\q!\la.\PP$, $\Q=\pp?\la.\Q$ and $\PR\neq\inact$, then $\pP\pp\PP\parN\pP\q\Q\parN\pP\pr\PR$ is $\participants$-excluded deadlock-free for any $\participants$, but $\participants$-excluded  lock-free only when $\pr\in\participants$.


\smallskip

The following example illustrates the notion of $\participants$-excluded  \Lc-freedom.

\begin{example}[Social media]\label{sm}\em
Let us consider a system describing a simplified social media situation.
Participant $\q$ is allowed to greet participant $\pp$ by sending a message   $\msg{hello}$.
Participant $\pp$ would like to reply to $\q$, but in order to do that she needs 
to be granted a higher communication level.
The task of granting permissions is performed by participant $\pu$ which,
upon $\pp$'s request ($\msg{req}$), decides -- according to some parameters -- whether the permission
is granted ($\msg{grtd}$) or denied ($\msg{dnd}$). Her decision is communicated to 
both $\pp$ and $\q$.
We assume that (for reusability motivation) $\pu$ is implemented in order to process an unbounded number of requests.
For what concerns $\pp$, however, once she is granted the higher communication level, she can return the greeting  to  
$\q$, so ending their interaction.
The above  system corresponds to the following session \\
\Cline{
\Nt \equiv \pP\pp{\PP}\parN \pP\q{\PQ}\parN \pP\pu{\U}
}

\noindent
where
$\PP=\q?\msg{hello}.\pu!\msg{req}.\pu?\Set{\msg{dnd}.\PP,\, \msg{grtd}.\q!\msg{hello}}$, 
$\PQ = \pp!\msg{hello}.\pu?\Set{\msg{dnd}.\PQ,\, \msg{grtd}.\pp?\msg{hello}}$  and 
$\U = \pp?\msg{req}.\pp!\Set{\msg{dnd}.\q!\msg{dnd}.\U,\, \msg{grtd}.\q!\msg{grtd}.\U}$. 
The session is  $\set\pu$-excluded  lock-free,  since, once $\pp$ has been granted the higher communication level, we get

\Cline{
\pP\pp{\inact}\parN \pP\q{\inact}\parN \pP\pu{\U}
}

\noindent
where $\pu$ is willing to interact but she will never be  able. 
This, however, should not be deemed a problem, since we are actually interested in
that the interactions between $\pp$ and $\q$ do proceed smoothly.
\finex
\end{example}


The behaviour of multiparty sessions can be disciplined by means of types.
Global types describe the 
conversation scenarios of multiparty  sessions,  possibly in a partial way.  

\begin{definition}[Global types]\label{def:gt}
{\em Global types} are defined by:

\Cline{\G\coDef\End\ \mid\ \Gc\pp\q\lambda\G i I}

\noindent
where $I\neq\emptyset$ and $\lambda_j\neq\lambda_h$ for  $j,h\in I$ and $j\neq h$.  
\end{definition}
 As for processes, we  allow  
 only {\em regular} global types. 
The type $\Gc\pp\q\lambda\G i I$ formalises a protocol where participant
$\pp$ must send  
to $\q$ a 
message  $\lambda_j$ for some $j\in I$ (and $\q$ must receive it) and
then, depending  on  which $\lambda_j$ was chosen by $\pp$, the protocol
continues as $\G_j$. We write $\pp\to\q:\lambda.\G$ when there is only one message.   We use $\End$ to denote the terminated protocol.  We shall omit writing trailing $\End$s in global types.

We define the {\em set of paths of a global type} $\G$, notation $\paths\G$, as the greatest set 
 of traces  such that:

\Cline{\paths\End=\set{\ee}\qquad\paths{\Gc\pp\q\lambda\G i I}=\bigcup_{i\in I}\set{\concat{\pp\lambda_i\q}\mypath\mid\mypath\in\paths{\G_i}}}

\noindent 
The {\em set of participants of a global type} is the set of participants of its paths:

\Cline{
\plays\G=\bigcup_{\mypath\in\paths\G}\plays\mypath
}

\noindent The regularity of global types ensures that  such  
sets of participants are  always   finite.

\medskip
 Boundedness is a property of global types that will enable us to get
$\participants$-excluded  \Lc-freedom  from typability.
This consists in requiring any participant of a global type to occur 
 either in all the paths or in no path   of any of its subterms  which are global types.
  Notably  this condition  is a form of fairness,   even  
  if it strongly differs from the notions of fairness discussed in~\cite{GHH21}, where fairness assumptions rule out computational paths. 
Technically,  we  shall  use  
the notions of {\em depth} and of  {\em boundedness} 
 as defined  below.
We denote by $\mypath[n]$ with $n \in \Nat$ the $n$-th  label  
in the path $\mypath$, where $1\le n \le \cardin{\mypath}$. 

\begin{definition}[Depth]\label{df:w}
Let $\G$ be a global type. 
For ${\mypath\in\paths{\G}}$  we define 

\Cline{\depth\mypath\pp = {\rm infimum} \{ n \mid \pp\in\plays{\mypath[n]} \} }

\noindent
and define $\depth\G\pp$, the \emph{depth} of $\pp$ in $\G$, as follows:

\Cline{\depth\G\pp = 
\begin{cases} 
{\rm supremum} \{ \depth\mypath\pp \mid \mypath \in \paths{\G} \}&  \text{if } \pp \in \plays{\G} \\ 
0 & \text{otherwise} 
\end{cases}
}

\end{definition}
 
 Note that $\depth\G\pp=0$ iff  $\pp \not\in \plays{\G}$. 
Moreover, if $\pp \in \plays{\G}$, but for some $\mypath \in \paths\G$ it is the case that $\pp\not\in\plays{\mypath[n]}$  for 
all $n \leq \cardin{\mypath}$,  then $\depth\mypath\pp = {\rm infimum}\, \emptyset = \infty$. 
Hence, if $\pp$ is a  participant of a global type $\G$ 
and there is some path in $\G$ where $\pp$ does not occur,  
then $\depth\G\pp= \infty$.

\begin{definition}[Boundedness]\label{def:bound}
A global type $\G$ is \emph{bounded} if  $\depth{\G'}\pp$ is finite
for all participants $\pp\in\plays{\G'}$ and  all types  
$\G'$ which occur in   $\G$. 
\end{definition}

Intuitively, this means that if $\pp \in \plays{\G'}$ for a type $\G'$ which occurs in $\G$, then the search for an interaction of the shape 
$\pp\lambda\q$ or $\q\lambda\pp$ along a path 
$\mypath \in \paths{\G'}$  terminates  
 (and recall that $\G'$ can be infinite, in which case $\G$ is such). Hence the name.
 
Example 2 of~\cite{BDL22} shows the necessity of considering all types occurring in a global type when defining boundedness and that also a finite global type can be unbounded.

%
%
%
%
%
%

Since global types are regular, the boundedness condition is decidable.  We shall allow only 
bounded global types in typing  sessions. 

\begin{example}[A global type for the social media example]
\label{ex:gtex} \em
The  intended  overall behaviour of the multiparty session $\Nt$
in Example~\ref{sm},  up to the point where the request of $\pp$ is possibly accepted by $\pu$,  is described by the following  global type $\G$.

\Cline{\begin{array}{lcl}\G = \q\to\pp{:}\msg{hello}.\,\pp\to\pu{:}\msg{req}.\,\pu\to\pp{:}
\left\{ \begin{array}{l}
                                               \msg{dnd}.\,\pu\to\q{:}\msg{dnd}.\,\G
                                 \\
                                 \msg{grtd}.\,\pu\to\q{:}\msg{grtd}.\,\pp\to\q{:}\msg{hello}
                           \end{array}\right. 
           \end{array}}

\noindent
  Typability of $\Nt$ with $\G$ -- by means of the type system defined in the next section -- will ensure (see Theorems~\ref{SR} and~\ref{SF} below)
that the  behaviours  of participants of $\G$, but $\pu$, will perfectly adhere to what $\G$ describes. 
\finex 
\end{example}

We  conclude  this section by defining  the standard LTS for global types.
 By means of such LTS we formalise the 
intended meaning of global types as overall (possibly partial) descriptions of sessions' behaviours.
 It will be used in the next section to prove the properties of Subject Reduction and Session Fidelity which, in our setting, will slightly differ from the standard ones~\cite{HYC08,Honda2016}.
 

\begin{definition}[LTS for Global Types]\label{ltsgt}
  The {\em labelled transition system (LTS) for global types} is
  specified by the  following axiom and rule: 
  
  \Cline{
    \begin{array}{c}
      \NamedRule{\rn{Ecomm}}{}{ \Gc\pp\q\lambda\G i I \stackred{\pp \lambda_j \q}{\G_j}}{j\in I}
      \\[3mm]
     \NamedRule{\rn{Icomm}}{  \G_i \stackred{\pp \lambda \q}\G'_i
      \quad
            \forall i\in I
      \qquad
      \set{\pp,\q}\cap\set{\pr,\ps}=\emptyset
      }{
     \Gc\pr\ps\lambda\G i I 
     \stackred{\pp \lambda \q}
      \Gc\pr\ps\lambda{\G'} i I  
      }{}
    \end{array}
 }
 
\end{definition}

 Axiom 
\rn{Ecomm} formalises the fact that, in a session exposing the behaviour
$\Gc\pp\q\lambda\G i I$,
there are participants $\pp$ and $\q$ ready to exchange  a message $\lambda_j$ for any $j\in I$,  the former as
sender and the latter as receiver. If such a communication is actually performed, the resulting 
session will expose the behaviour  $\G_j$. 

Rule \rn{Icomm} makes sense since, in a global type
$\Gc\pr\ps\lambda\G i I $,  communications 
involving
participants $\pp$ and $\q$, ready to interact with each other
uniformly in all branches, can  be performed  
if neither of them is involved in a previous interaction between $\pr$ and $\ps$.
In this case, the interaction between $\pp$ and $\q$ is independent of the
choice of $\pr$, and may be executed before it.

\section{Type System and its Properties}\label{ts}

 As in~\cite{BDL22,CDG22,DGD22},   our   type assignment allows for a simple 
treatment of many technical issues, by avoiding  projections,  local types 
and subtyping~\cite{HYC08,Honda2016}.  
The novelty of the type system we present in this section  with respect to those
in~\cite{BDL22,CDG22,DGD22}
is that the judgments are parametrised by a set $\participants$ of participants. These are the participants whose  \Lc-freedom  we do not care about. 
The simplicity of our calculus allows us to formulate a type system deriving directly global types for multiparty sessions,  i.e. judgments of the form $\tyng\G\Nt\participants$  (where $\G$ is bounded).  Here and in the following the double line indicates that the rules are interpreted coinductively~\cite[Chapter 21]{pier02}.  

\begin{definition}[Type system]\label{def:type-system} 
The type system $\vdash_\participants$ is defined by the following axiom and  rules, 
 where sessions are considered modulo structural  congruence:  

\smallskip

\Cline{\NamedCoRule{\rn{\End}}{\tyng\End{\pP\pp\inact}\emptyset}{}{}{}}

\smallskip

\Cline{\NamedCoRule{\rn{Comm}}
{\mbox{$\begin{array}{c}\tyng{\G_i}{\pP\pp{\PP_i}}{\participants_i}\parN\pP\q{\Q_i}\parN\Nt\\[1mm]
(\plays{\G_i}\cup\participants_i)\setminus\set{\pp,\q}=\plays\Nt\quad\forall i\in I \\
\ \end{array}$}}
{\tyng{\G}{\pP\pp{\Pou\q\lambda\PP i I}\parN\pP\q{\Pin\pp\lambda\Q j J}\parN\Nt}
{\participants}}
{
\begin{array}{c}
\G=\Gc\pp\q\lambda\G i I\quad
\G\mathrm{\ is\ bounded}\\
\participants=\bigcup_{i\in I}\!\participants_i\quad
I\subseteq J
\end{array}}}

\smallskip

\Cline{\NamedCoRule{\rn{\sw}}
{\tyng{\G}{\Nt_1}{\participants_1}}
{\tyng{\G}{\Nt_1\parN\Nt_2}{\participants_1\cup\participants_2}}
{ \participants_2 = \plays{\Nt_2} \neq\emptyset 
}{}}

\end{definition}

Axiom \rn{\End} simply states that the null session has  the 
behaviour described by 
$\End$. 
In the null session there is obviously no participant whose  \Lc-freedom  we do not care about,
hence the $\emptyset$ subscript. An alternative and sound version of this axiom could be 
$\tyng\End{\Nt}{\plays{\Nt}}$ for any $\Nt$. Such a judgment however can be easily derived using Rule  \rn{\sw}.

Rule \rn{Comm} just adds simultaneous communications to global types and to corresponding processes inside sessions. 
Since  the set $\participants_i$ contains the ignored participants 
in branch $i$, 
the possibly non lock-free participants in the conclusion  must  be   
$\bigcup_{i\in I}\!\participants_i$.
 Note that Rule \rn{Comm} allows more inputs 
 than corresponding outputs, in agreement with the condition in Rule \rn{Comm-T}.
 It also allows more branches 
in the input process than in the global type, just mimicking the subtyping for session types~\cite{DH12}. 
 Instead, the number of branches in the output process and the global type must be the same. 
This does not 
restrict typability, 
while it improves Session Fidelity 
(by arguing, respectively, as in~\cite{BDLT21} and~\cite{BDL22}). 
The condition  $(\plays{\G_i}\cup\participants_i)\setminus\set{\pp,\q}=\plays\Nt$  for all $i\in I$ 
ensures that  the participants in the session are exactly those we keep track of  either in $\G$ and/or in $\participants$. 
This condition prevents,
for example, to derive   $\tyng{\G}{\pP\pp{\PP}\parN\pP\q{\PQ}\parN\pP\pr\R}{\emptyset}$,
where $\G=\pp\to\q:\lambda.\G$,  $\PP=\q!\lambda.\PP$, $\PQ =\pp?\lambda.\PQ$ and  
$\R\neq\inact$  is arbitrary.
 Note that, instead, it is possible to derive  $\tyng{\G}{\pP\pp{\PP}\parN\pP\q{\PQ}\parN\pP\pr\R}{\Set{\pr}}$  with  $\R\neq\inact$  arbitrary. 
 Lock-freedom can be  ensured only for the participants of  $\Nt$ 
not belonging to  $\participants$. 

Rule  \rn{\sw} enables to type check just a sub session as far as we do not care about
the  \Lc-freedom  of the participants of the rest of the session. We keep track of such participants
in the subscript of the entailment symbol. 
The condition $\participants_2\neq\emptyset$  forbids infinite applications of this rule.  This condition allows us to use coinduction on typing derivations.  
 
 Sessions are considered modulo structural  congruence  
 in typing rules following~\cite{BDL22,CDG22,DGD22}. Clearly this could be avoided by adding an obvious typing rule, but we prefer to have a lighter type system.

 The regularity of processes and global types ensures the decidability of type checking.

\begin{figure}[b]
$
\mbox{\Large $\DD$} =\
\prooftree
         \prooftree
                  \prooftree
                         \prooftree
                            \mbox{\Large $\DD$}
                         \Justifies
                              \tyng{\G_3}{\pP\pp{\PP}\parN \pP\q{\PQ_1}\parN \pP\pu{\q!\msg{dnd}.\PU}}{\Set{\pu}}
                         \endprooftree     
                         \qquad
                         \prooftree
                                 \prooftree
                                          \prooftree
                                                 \prooftree
                                                 \Justifies
                                                     \tyng{\End}{\pP\pp{\inact}\parN \pP\q{\inact}}{\emptyset}
                                                 \using \rn{$\End$}
                                                 \endprooftree
                                           \Justifies
                                                   \tyng{\G_5}{\pP\pp{\PP_2}\parN \pP\q{\PQ_2}}{\emptyset}
                                           \endprooftree  
                                 \Justifies
                                     \tyng{\G_5}{\pP\pp{\PP_2}\parN \pP\q{\PQ_2}\parN \pP\pu{\PU}}{\Set{\pu}}
                                  \using \rn{\sw}   
                                  \endprooftree   
                         \Justifies
                             \tyng{\G_4}{\pP\pp{\PP_2}\parN \pP\q{\PQ_1}\parN \pP\pu{\q!\msg{grtd}.\PU}}{\Set{\pu}}
                          \endprooftree   
                  \Justifies
                     \tyng{\G_2}{\pP\pp{\pu?\Set{\msg{dnd}.\PP,\, \msg{grtd}.\PP_2}}\parN \pP\q{\PQ_1}\parN \pP\pu{\PU_1}}{\Set{\pu}}
                  \endprooftree
         \Justifies
             \tyng{\G_1}{\pP\pp{\PP_1}\parN \pP\q{\PQ_1}\parN \pP\pu{\PU}}{\Set{\pu}}
          \endprooftree
\Justifies
          \tyng{\G}{\pP\pp{\PP}\parN \pP\q{\PQ}\parN \pP\pu{\PU}}{\Set{\pu}}
\endprooftree
$
\caption{A type derivation for  the social media.}\label{fsm}
 \end{figure}

\begin{example}[Typing of the social media]  
\label{ex:typingwe} \em

Let $\PP, \Q$ and $\PU$ be defined as in Example~\ref{sm} and

\Cline{\begin{array}{c}
\PP_1=\pu!\msg{req}.\pu?\Set{\msg{dnd}.\PP,\, \msg{grtd}.\PP_2},\, 
\PP_2=\q!\msg{hello},\,
\PQ_1 = \pu?\Set{\msg{dnd}.\PQ,\, \msg{grtd}.\PQ_2},\,
\PQ_2 = \pp?\msg{hello}\\
\PU_1 = \pp!\Set{\msg{dnd}.\q!\msg{dnd}.\PU,\, \msg{grtd}.\q!\msg{grtd}.\PU}\end{array}} 

Moreover, let $\G$  be defined as in Example~\ref{ex:gtex} and  
$\G_1 = \pp\to\pu{:}\msg{req}.\, \G_2
$ 

\Cline{ 
\G_2 = \pu\to\pp{:}
\left\{ \begin{array}{l}
                                               \msg{dnd}.\, \G_3
                                 \\
                                 \msg{grtd}.\,\G_4
                           \end{array}\right. 
 \quad
\G_3=\pu\to\q{:}\msg{dnd}.\,\G \quad
\G_4=\pu\to\q{:}\msg{grtd}.\,\G_5 \quad
\G_5 =\pp\to\q{:}\msg{hello}
} 

Figure~\ref{fsm} shows a derivation of the global type $\G$ of Example~\ref{ex:gtex} for the multiparty session $\Nt$ of  Example~\ref{sm}.  
The missing rule names are all
\rn{Comm}.  
 Note how in the leftmost branch of the 
 derivation it is possible to get $\Set{\pu}$ as subscript without recurring to Rule
\rn{\sw} thanks to the  infiniteness 
of the branch. 
For the same motivation, in case we had 
$\PP=\q?\msg{hello}.\pu!\msg{req}.\pu?\Set{\msg{dnd}.\PP,\, \msg{grtd}.\PP'}$
 with $\PP'= \q!\msg{hello}.\PP'$  and   
 $\PQ = \pp!\msg{hello}.\pu?\Set{\msg{dnd}.\PQ,\, \msg{grtd}.\PQ'}$ with $\PQ' = \pp?\msg{hello}.\PQ'$
(namely in case $\pp$ and $\q$ kept on indefinitely exchanging $\msg{hello}$ messages after 
receiving the \msg{grtd} message)
the whole resulting session  would be typable without recurring to Rule \rn{\sw}.   
\finex 
\end{example}

 We note that session participants are of three different kinds in a typing judgment:
\begin{enumerate}
\item the lock-free participants which behave as pointed out by the global type; these participants occur in the global type but do not belong to the set of ignored participants;
\item the participants which ``partially'' behave as pointed out by the global type and can  get  
stuck; these participants occur in the global type and belong to the set of ignored participants;
\item the participants which behave in an  unpredictable 
way; these participants do not occur in the global type but belong to the set of ignored participants.
\end{enumerate}

We  observe  also  that $\tyng{}{}{\emptyset}$ coincides with 
the typing relation of~\cite{BDL22}.


\medskip

In  the remainder 
of  this section we will show the main properties of our type system, i.e. Subject Reduction, Session Fidelity and $\participants$-excluded Lock-freedom. We start with some lemmas which are handy for the  subsequent  
proofs.  All proofs are by  coinduction on $\tyng\G\Nt\participants$ and by cases on the last applied
rule. 

The first lemma  states 
that, when $\tyng\G\Nt\participants$, all participants of $\Nt$ must be participants of $\G$ and/or must belong to the set 
$\participants$.

\begin{lemma}\label{lem:triv}
$\tyng\G\Nt\participants$ implies $\plays\G\cup\participants=\plays\Nt$. 
\end{lemma}
\begin{proof}  
%
{\em Rule} \rn{Comm}. Immediate by the condition  $(\plays{\G_i}\cup\participants_i)\setminus\set{\pp,\q}=\plays\Nt$ for all $i\in I$.   

{\em Rule} \rn{\sw}. In such a case, $\Nt\equiv\Nt_1\parN\Nt_2$
and $\participants = \participants_1\cup\participants_2$. 
By coinduction we get  $\plays\G\cup\participants_1=\plays{\Nt_1}$.  We hence get the thesis 
by  the condition $\participants_2 = \plays{\Nt_2}$.
\qed
\end{proof}

\medskip
 By the above lemma, from $\tyng\G\Nt\participants$ it is immediate to get also that $\pp\in\plays\G$ implies $\pp\in\plays\Nt$ and that
$\pp\in\plays\Nt$ and $\pp\not\in\participants$ imply
$\pp\in\plays\G$.  

\smallskip

 The process of a participant which does not occur in the global type can be freely replaced, since typing  ensures  nothing about the behaviour of this participant.

\begin{lemma}\label{a} 
If  $\tyng\G{\pP\pp\PP\parN\Nt}\participants$   and $\PP\neq\inact$ and 
$\pp\not\in\plays{\G}$, 
then $\tyng\G{\pP\pp{\PP'}\parN\Nt}{\participants'}$ with $\participants'\subseteq \participants$ for an  arbitrary $\PP'$.
\end{lemma}
\begin{proof}
%
{\em Rule} \rn{Comm}. Then $\G=\Gc\q\pr\lambda\G i I$ and $\Nt\equiv\pP\q  {\Pou\pr\lambda\Q i I} \parN\pP \pr {\Pin\q\lambda\PR j J }\parN\Nt_0$   and 
$I\subseteq J$ and  $\tyng{\G_i}{\pP\pp\PP\parN\pP\q{\Q_i}\parN\pP\pr{\PR_i}\parN\Nt_0}{\participants_i}$ for all $i \in I$ with $\participants=\bigcup_{i\in I}\participants_i$. By coinduction  we get  \linebreak $\tyng{\G_i}{\pP\pp{\PP'}\parN\pP\q{\Q_i}\parN\pP\pr{\PR_i}\parN\Nt_0}{\participants_i'}$ with $\participants_i'\subseteq \participants_i$ for all $i \in I$ for  an  arbitrary $\PP'$. We conclude using Rule  \rn{Comm}.

{\em Rule} \rn{\sw}. Then $\participants=\participants_1\cup\participants_2$  and   
$\pP\pp\PP\parN\Nt\equiv\Nt_1\parN\Nt_2$ and $\tyng\G{\Nt_1}{\participants_1}$. If $\Nt_1\equiv\pP\pp\PP\parN\Nt_1'$ by coinduction $\tyng\G{\pP\pp{\PP'}\parN\Nt_1'}{\participants_1'}$  with $\participants_1'\subseteq \participants_1$ for arbitrary $\PP'$
and we conclude using Rule  \rn{\sw}. If $\Nt_2\equiv\pP\pp\PP\parN\Nt_2'$ we can apply Rule  \rn{\sw} to  $\tyng\G{\Nt_1}{\participants_1}$ and $\pP\pp{\PP'}\parN\Nt_2'$  for arbitrary $\PP'$.
\qed
\end{proof}

\medskip
\noindent
 Note  that  in previous lemma  $\participants'=\participants$ unless $\PP'=\inact$ and in this case $\participants'\cup\set\pp=\participants$. 

\smallskip

If $\pP\pp{\Pb\q\lambda\PP i I} \in \Nt$ we say that $\q$ is the {\em top partner} of $\pp$ and we write $\tp\Nt\pp=\q$.  Note  that we can have $\q\not\in\plays{\Nt}$ or  $\tp\Nt\q\not=\pp$. 
For example, if $\Nt\equiv\pP\pp{\Pb\q\lambda\PP i I}\parN\pP\q{\Pb\pr{\lambda'}\Q j J}$,
we have $\tp\Nt\pp=\q$ and $\tp\Nt\q=\pr\neq\pp$. 

 Typing  ensures  that if a participant occurs in a global type then also her top partner occurs in the global type.

\begin{lemma}\label{d} 
If $\tyng\G\Nt\participants$ and $\pp\in\plays{\G}$, then $\tp\Nt\pp\in\plays{\G}$.
\end{lemma}
\begin{proof} By Lemma~\ref{lem:triv}  and   $\pp\in\plays{\G}$ we have that $\pp\in\plays{\Nt}$ and then $\tp\Nt\pp$ is defined. So, let $\tp\Nt\pp=\q$. 

{\em Rule} \rn{Comm}. Then $\G=\Gc\pr\ps\lambda\G i I$ and $\Nt\equiv\pP\pr  {\Pou\ps\lambda\PR i I} \parN\pP \ps {\Pin\pr\lambda\PS j J }\parN\Nt_0$  with $I\subseteq J$ and  $\tyng{\G_i}{\pP\pr{\PR_i}\parN\pP\ps{\PS_i}\parN\Nt_0}{\participants_i}$ for all $i \in I$ with $\participants=\bigcup_{i\in I}\participants_i$. If $\pp\in\set{\pr,\ps}$, then $\set{\pp,\q}=\set{\pr,\ps}$ and we are done. Otherwise $\tp\Nt\pp=\q$ implies $\tp{\pP\pr{\PR_i}\parN\pP\ps{\PS_i}\parN\Nt_0}\pp=\q$
for all $i \in I$. Moreover
$\pp\in\plays\G$ implies $\pp\in\plays{\G_i}$ for all $i\in I$, since $\G$ is bounded. By coinduction we get $\q\in\plays{\G_i}$ for all $i\in I$. We conclude $\q\in\plays\G$.

{\em Rule} \rn{\sw}. Then $\participants=\participants_1\cup\participants_2$  and  $\Nt\equiv\Nt_1\parN\Nt_2$ and $\tyng\G{\Nt_1}{\participants_1}$. Since  by Lemma~\ref{lem:triv}  
$\pp\in\plays{\G}$ implies $\pp\in\plays{\Nt_1}$ we have $\tp{\Nt_1}\pp=\q$. We get by coinduction $\q\in\plays\G$. \qed
\end{proof}

\medskip

  In our particular setting, what Subject Reduction  ensures depends on 
which participants we consider (unlike its standard version, e.g. in~\cite{BDLT21} and~\cite{BDL22}).
In particular, it ensures that, when the involved participants occur in the global types,
 the transitions of well-typed sessions are
mimicked by those of global types
 (namely they  proceed  as prescribed by the global type). 
Otherwise the reduced session can be typed by the same global type.  Key for this proof is Lemma~\ref{d}, which ensures that the communicating participants either both occur or both  do not occur in the global type.

\begin{theorem}[Subject Reduction]\label{SR}
Let $\tyng\G\Nt\participants$ and $\Nt  \stackred{\pp \lambda \q} \Nt'$.
\begin{enumerate}[i)]
\item\label{SR1}
  If  $\set{\pp,\q}\subseteq\plays\G$, then 
  $\G \stackred{\pp \lambda \q}\G'$ and $\tyng{\G'}{\Nt'}{\participants'}$ with $\participants'\subseteq\participants$.
\item\label{SR2}
  If  $\pp,\q\not\in\plays\G$, then
  $\tyng{\G}{\Nt'}{\participants'}$ with $\participants'\subseteq\participants$.
\end{enumerate}
\end{theorem}
\begin{proof}
From $\Nt  \stackred{\pp \lambda \q} \Nt'$ we get $\Nt\equiv\pP\pp  {\Pou\q\lambda\PP i I} \parN\pP \q {\Pin\pp\lambda\Q j J }\parN\Nt_0$  with $I\subseteq J$ 
and \linebreak $\Nt'\equiv\pP\pp  {\PP_l} \parN\pP \q {\Q_l}\parN\Nt_0$ and $\lambda=\lambda_l$ for some $l\in I$.  Note  that Lemma~\ref{d} implies either $\set{\pp,\q}\subseteq\plays\G$ or 
 $\pp,\q\not\in\plays\G$.\\
{\bf \ref{SR2}}. 
 In this case $\tyng\G\Nt\participants$ implies $\tyng\G{\pP\pp  {\PP_l} \parN\pP \q {\Q_l}\parN\Nt_0}{\participants'}$ with $\participants'\subseteq\participants$ by Lemma~\ref{a}.\\
{\bf \ref{SR1}}. 
The proof is by  coinduction on $\tyng\G{\pP\pp\PP\parN\Nt}\participants$ and by cases on the last applied rule.  

{\em Rule} \rn{Comm}. We get $\G=\Gc\pr\ps{\lambda'}\G h H$ and  $\Nt\equiv\pP\pr  {\Pou\ps{\lambda'}\R h H} \parN\pP \ps {\Pin\pr{\lambda'}S k K }\parN\Nt_1$  and $H\subseteq K$ and $\tyng{\G_h}{\pP\pr{\R_h}\parN\pP\ps{S_h}\parN\Nt_1}{\participants_h}$ for all $h \in H$ with $\participants=\bigcup_{h\in H}\participants_h$. If with $\pp=\pr$ and $\q=\ps$, then $I=H$, $J=K$ and $\lambda_i=\lambda_i'$ for all $i\in I$. We conclude  $\G \stackred{\pp \lambda \q}\G_l$ and $\tyng{\G_l}{\Nt'}{\participants_l}$. Otherwise 
$\set{\pp,\q}\cap\set{\pr,\ps}=\emptyset$, which
implies  $
\pP\pr  {\R_h} \parN\pP \ps {S_h }\parN\Nt_1\stackred{\pp \lambda \q} \pP\pr  {\R_h} \parN\pP \ps {S_h }\parN\Nt'_1$ for all $h \in H$ by Fact~\ref{fact:arb}. 
Moreover $\set{\pp,\q}\subseteq\plays{\G_h}$ for all $h \in H$ since $\G$ is bounded. 
By  coinduction  we get $\G_h\stackred{\pp \lambda \q}\G_h'$ and $ \tyng{\G'_h}{\pP\pr  {\R_h} \parN\pP \ps {S_h }\parN\Nt'_1}{\participants'_h}$  for some $\G'_h$  and  ${\participants'_h}\subseteq{\participants_h}$ and  for all $h\in H$. 
We conclude $\G\stackred{\pp \lambda \q}\Gc\pr\ps{\lambda'}{\G'} h H$ using Rule \rn{Icomm} and $ \tyng{\Gc\pr\ps{\lambda'}{\G'} h H}{\Nt'}{\participants'}$ with $\participants'=\bigcup_{h\in H}\participants_h'$ using Rule \rn{Comm}.

{\em Rule} \rn{\sw}. In this case $\Nt\equiv\Nt_1\parN\Nt_2$ and $\participants=\participants_1\cup\participants_2$ and $\tyng\G{\Nt_1}{\participants_1}$. From $\set{\pp,\q}\subseteq\plays\G$,  by Lemma~\ref{lem:triv}  we get $\set{\pp,\q}\subseteq\plays{\Nt_1}$ which implies $\Nt_1  \stackred{\pp \lambda \q} \Nt'_1$. 
By coinduction we get  $\G\stackred{\pp \lambda \q}\G'$ and $\tyng{\G'}{\Nt_1'}{\participants_1'}$ with ${\participants'_1}\subseteq{\participants_1}$. We  conclude using Rule \rn{\sw}, since by construction $\Nt'\equiv\Nt_1'\parN\Nt_2$.\,\qed\linebreak
\end{proof}


 We   note  that Subject Reduction, as formulated in previous theorem, fails if we allow unbounded global types.  
 Let $\G=\pp\to\q{:}\set{\lambda_1.\pr\to\ps{:}\lambda,\lambda_2.\G}$ and $\Nt\equiv\pP\pp\PP\parN\pP\q\Q\parN\pP\pr{\ps!\lambda}\parN\pP\ps{\pr?\lambda}$ and  
 $\PP=\q!\Set{\lambda_1,\,\lambda_2.\PP}$ and 
 $\Q=\pp?\Set{\lambda_1,\,\lambda_2.\Q}$.  
 Then we have  $\tyng\G\Nt\emptyset$ and $\Nt\stackred{\pr \lambda \ps}\pP\pp\PP\parN\pP\q\Q$, but there is no transition labelled $\pr \lambda \ps$ starting from $\G$.  Note  that the session $\Nt$ can be typed,  still with the $\emptyset$ subscript,  by the bounded global type $\G'=\pr\to\ps{:}\lambda.\pp\to\q{:}\set{\lambda_1,\lambda_2.\G'}$.

\smallskip

Session  Fidelity  ensures that the communications in a session typed by a global type proceed   
 at least as prescribed by the global  type.

\begin{theorem}[Session Fidelity]\label{SF}
Let $\tyng\G\Nt\participants$ and
$\G \stackred{\pp \lambda \q} \G'$.
 Then  $\Nt \stackred{\pp \lambda \q} \Nt'$ and $\tyng{\G'}{\Nt'}{\participants'}$ with $\participants'\subseteq\participants$. 
\end{theorem}
\begin{proof} The proof is by coinduction on the derivation of $\tyng\G\Nt\participants$ and by cases on the last applied rule.

{\em Rule} \rn{Comm}.
The proof is by induction on  the number  $t$  of 
transition  rules used to derive $\G \stackred{\pp \lambda \q} \G'$.\\
{\em Case $t=1$.} Then $\G \stackred{\pp \lambda \q} \G'$ is  the Axiom 
\rn{Ecomm}  and   $\G=\Gc\pp\q\lambda\G i I$, where 
$\lambda=\lambda_l$ and $\G'=\G_l$ with $l\in I$.
We get
 $\Nt\equiv\pP\pp  {\Pou\q\lambda\PP i I} \parN\pP \q {\Pin\pp\lambda\Q j J }\parN\Nt_0$ 
 with $I\subseteq J$ and $\tyng{\G_i}{\pP\pp{\PP_i}\parN\pP\q{\Q_i}\parN\Nt_0}{\participants_i}$ for all $i \in I$ with $\participants=\bigcup_{i\in I}\participants_i$.
   Then  we conclude 
 $\Nt\stackred{\pp \lambda \q} \pP\pp{\PP_l}\parN\pP\q{\Q_l}\parN\Nt_0$  by Rule \rn{Comm-T}  and  $ \tyng{\G_l}{\pP\pp{\PP_l}\parN\pP\q{\Q_l}\parN\Nt_0}{\participants_l}$.\\ 
 {\em Case $t>1$.} Then $\G \stackred{\pp \lambda \q} \G'$ is the conclusion of  Rule \rn{Icomm}. Moreover 
$\G=\Gc\pr\ps{\lambda'}\G h H$ and $\G'=\Gc\pr\ps{\lambda'}{\G'} h H$ and  $\G_h\stackred{\pp \lambda \q}\G_h'$ for all $h \in H$ and $\set{\pp,\q}\cap\set{\pr,\ps}=\emptyset$. We get

 \Cline{\Nt\equiv\pP\pr  {\Pou\ps{\lambda'}\R h H} \parN\pP \ps {\Pin\pr{\lambda'}S k K}\parN\Nt_1} \noindent with $H\subseteq K\text{ and }\tyng{\G_h}{\pP\pr{\R_h}\parN\pP\ps{S_h}\parN\Nt_1}{\participants_h}$ for all $h \in H$ with $\participants=\bigcup_{h\in H}\participants_h$.
 
 \noindent
By induction  $\pP\pr  {\R_h} \parN\pP \ps {S_h }\parN\Nt_1\stackred{\pp \lambda \q} \Nt'_h$ and $\tyng{\G'_h}{\Nt'_h}{\participants'_h}$ with ${\participants_h'}\subseteq{\participants_h}$ for all $h\in H$. 
The condition $\set{\pp,\q}\cap\set{\pr,\ps}=\emptyset$ ensures that the reduction  $\pP\pr  {\R_h} \parN\pP \ps {S_h }\parN\Nt_1\stackred{\pp \lambda \q} \Nt'_h$ does not modify the processes of participants $\pr$ and $\ps$. Moreover the processes of participants $\pp$ and $\q$ are the same in  $ \Nt'_h$ for all $h\in H$. This implies  $\Nt'_h\equiv\pP\pr{\R_h}\parN\pP\ps{S_h}\parN\Nt''$ for all $h\in H$ and some $\Nt''$. 
We conclude $\Nt\stackred{\pp \lambda \q}\Nt'$  where $\Nt'=\pP\pr  {\Pou\ps{\lambda'}\R h H} \parN\pP \ps {\Pin\pr{\lambda'}S k K }\parN\Nt''$  using Rule \rn{Comm-T} and $ \tyng{\G'}{\Nt'}{\participants'}$ with $\participants'=\bigcup_{h\in H}\participants_h'$ using Rule \rn{Comm}.

{\em Rule} \rn{\sw}. In this case $\Nt\equiv\Nt_1\parN\Nt_2$ and $\participants=\participants_1\cup\participants_2$ and $\tyng\G{\Nt_1}{\participants_1}$. By coinduction $\Nt_1 \stackred{\pp \lambda \q} \Nt'_1$ 
and $\tyng{\G'}{\Nt'_1}{\participants_1'}$ with $\participants_1'\subseteq\participants_1$. Then $\Nt_1 \parN\Nt_2\stackred{\pp \lambda \q} \Nt'_1\parN\Nt_2$ and $\tyng{\G'}{\Nt'_1\parN\Nt_2}{\participants_1'\cup\participants_2}$ using Rule \rn{\sw}.
\qed
\end{proof}

\medskip

 We  can show that  typability ensures  $\participants$-excluded Lock-freedom.
This
follows from   Subject Reduction and  Session Fidelity  thanks to the boundedness condition.

\begin{theorem}[$\participants$-excluded Lock-freedom]\label{LF}
If $\tyng\G\Nt\participants$,  
  then $\Nt$ is $\participants$-excluded lock-free.
\end{theorem} 
\begin{proof} 
 By Subject Reduction it is enough to prove that  in well-typed sessions
no active participant  not belonging to $\participants$  is  prevented to progress.  
So, let $\pp\in\plays\Nt$ such that $\pp\not\in\participants$.
 By Lemma~\ref{lem:triv}  we have that  $\pp\in\plays\Nt$ and $\pp\not\in\participants$ imply  $\pp\in\plays\G$.
 We proceed now by induction on $d=\depth\G\pp$.\\
 If $d=1$ then either $\G=\Gc\pp\q\lambda\G i I$ or
 $\G=\Gc\q\pp\lambda\G i I$ and $\G\stackred{\Lambda} \G'$ with $\pp\in\plays\Lambda$ by Axiom \rn{Ecomm}. 
 Then $\Nt \stackred{\Lambda}\Nt'$ by Theorem~\ref{SF}.\\
 If $d>1$ then $\G=\Gc\q\pr\lambda\G i I$ with $\pp\not\in\set{\q,\pr}$
 and $\G \stackred{\q \lambda_i \pr} \G_i$ for all $i\in I$ by
  Axiom \rn{Ecomm}. 
 Induction
 applies since $\depth\G\pp>\depth{\G_i}\pp$  for all $i\in I$.  Then,  for all $i\in I$, we get  $\G_i\stackred{\mypath_i\cdot\Lambda_i}\G_i'$  for some $\mypath_i$, $\Lambda_i$ with $\pp\in\plays{\Lambda_i}$. 
 This implies $\G\stackred{\sigma_i'}\G_i'$ where $\sigma_i'=\q \lambda_i \pr\cdot\mypath_i\cdot\Lambda_i$ for all $i\in I$. We conclude $\Nt\stackred{\sigma_i'}\Nt_i'$ for all $i\in I$ by Theorem~\ref{SF}.  
 \qed
 \end{proof}
 
 \medskip
  The following example shows that partial typing allows to type sessions which require unbounded global types in standard type systems~\cite{Honda2016}. 
 
 \begin{example}[Buyer-Seller-Carrier]
{\em
Let us consider the following session  (from~\cite[Sect.1]{PD22})

\Cline{
\Nt = \pP\pb{\PB}\parN \pP\ps{\PS}\parN \pP\pc{\ps?\msg{ship}}
}

\noindent
where $\PB=\ps!\Set{\msg{add}.\PB,\, \msg{pay}}$ and 
$\PS=\pb?\Set{\msg{add}.\PS,\, \msg{pay}.\pc!\msg{ship} }$

Such a session implements a system where
a $\pb$uyer can keep on adding goods - sold by a $\ps$eller -  in his shopping cart an unbounded number of times, until he decides to  pay for  
the shopping cart's content. 
In the latter case, the $\ps$eller informs the $\pc$arrier for the shipment.
Session $\Nt$ is obviously non lock-free, since participant $\pc$ would not be able to progress in
case $\ps$ be a seriously disturbed shopaholic who never  stops 
adding goods in his cart.
In fact $\Nt$ cannot be typed in $\tyng{}{}{\emptyset}$ (which corresponds to the type system of
\cite{BDL22}).

In this scenario participant $\pb$ is a client, whereas  $\ps$ and  $\pc$ are part of the service
used by $\pb$. It is hence natural to look at  $\Nt$ with a bias towards the client, the one whose
good property have to be  ensured. 
As a matter of fact it is possible to  ensure the  Lock-freedom  of  $\pb$, namely the 
$\Set{\ps,\pc}$-excluded lock freedom of $\Nt$, by  deriving 
$\tyng{\G}{\Nt}{\Set{\ps,\pc}}$,
where $\G = \pb\to\ps{:}\Set{\msg{add}.\,\G,\, \msg{pay}
 }$, as follows 


\Cline{\mbox{\Large $\DD$} \quad=\quad
                  \prooftree
                            \mbox{\Large $\DD$}
                         \qquad\qquad\qquad
                         \prooftree
                                                 \prooftree
                                                 \Justifies
                                                     \tyng{\End}{\pP\pb{\inact}}{\emptyset}
                                                 \using \rn{$\End$}
                                                 \endprooftree
                                           \Justifies
                                                   \tyng{\End}{\pP\pb\inact\parN\pP\ps{\pc!\msg{ship}}\parN \pP\pc{\ps?\msg{ship}}}{\Set{\ps,\pc}}
                                                    \using \rn{Weak}  
                                           \endprooftree  
                  \Justifies
                     \tyng{\G}{\pP\pb{\ps!\Set{\msg{add}.\PB,\, \msg{pay}}}\parN \pP\ps{\pb?\Set{\msg{add}.\PS,\, \msg{pay}.\pc!\msg{ship} }}\parN \pP\pc{\ps?\msg{ship}}}{\Set{\ps,\pc}}
                     \using \rn{Comm}
                  \endprooftree}
\medskip
\noindent 
 Note that $\G$  is a bounded global type. 
%
%
\finex
}
 \end{example}

\section{Type Inference}\label{ti}

 In  our type system each session can be trivially typed by the $\End$ type just applying  Axiom \rn{$\End$}  and Rule \rn{\sw}:

\Cline{\tyng\End\Nt{\plays\Nt}}

\noindent
Clearly, this typing does not provide any information on $\Nt$. 
We are interested in more informative typings,  i.e. with minimum sets of ignored participants,   if any. In this  section, we will describe an algorithm to infer global types  and sets of participants  from  sessions, 
proving  also its soundness and completeness  with respect to  our type system. 
 In particular,  the algorithm   applied to a  session 
$\Nt$  returns  
all and only those global types which can be  assigned to $\Nt$ with derivations indexed by  suitable 
sets of participants.   
 Note that, since  derivations indexed by the same or different sets of participants can assign different global types to a  session,  
 the algorithm needs to be non-deterministic in order to be complete. 

The first step towards defining  such an algorithm is the introduction of a finite representation for global types.
Since global types are regular terms,  they can be represented,  by results in~\cite{AdamekMV06,Cour83}, as finite systems of regular syntactic equations  
formally defined below. 

We begin by defining a \emph{global type pattern} as
a finite term generated by the following grammar: 

\Cline{
\Gpat ::= \End \mid \Gc\pp\q\lambda\Gpat i I \mid \X
}

\noindent
where $\X$ is a  type  variable taken from a countably infinite set. 
  We   denote by $\vars\Gpat$ the set of  type variables occurring in $\Gpat$. 
  
 We also need to compute sets of participants, so we define {\em  p-set patterns}  by:

\Cline{
\Ppat ::= \participants \mid \x \mid \Ppat\cup\Ppat
}

\noindent
where $\x$ is a  p-set  variable taken from a countably infinite set and $\participants$ can be any finite set of participants.
We   denote by $\vars\Ppat$ the set of   p-set     variables occurring in $\Ppat$. 

We use $\chi$ to range over type and p-set variables.  

A \emph{substitution} $\theta$ is a finite partial map from  type variables to global types and from p-set variables to sets of participants.   
  We   denote 
by $\theta \sbtplus \sigma$ the union of two substitutions such that  $\theta(\chi) = \sigma(\chi)$, for all $\chi \in \dom\theta\cap\dom\sigma$, 
and by $\Gpat\theta$ (resp. $\Ppat\theta$) the application of $\theta$ to $\Gpat$ (resp. $\Ppat$).
   We define   $\theta \sbtord \sigma$ if $\dom\theta\subseteq\dom\sigma$ and 
$\theta(\chi) = \sigma(\chi)$, for all $\chi \in \dom\theta$.  
 Note that, if $\vars\Gpat\subseteq\dom\theta$, then $\Gpat\theta$ is a global type and if $\vars\Ppat\subseteq\dom\theta$, then $\Ppat\theta$ is  a set of participants.

A \emph{type equation} has shape $\agteq\X\Gpat$ and a \emph{(regular) system of  type  equations} $\eqsys$ is a finite set of equations such that 
$\agteq\X{\Gpat_1}$ and $\agteq\X{\Gpat_2}\in\eqsys$ imply $\Gpat_1 = \Gpat_2$. 
We   denote  by $\dom\eqsys$ the set $\{\X\mid \agteq\X\Gpat \in \eqsys \}$ and by 
  $\vars\eqsys$ the set $\bigcup \{ \vars\Gpat\cup\{\X\} \mid \agteq\X\Gpat\in\eqsys\}$.
A \emph{solution} of a system $\eqsys$ is a substitution $\theta$ such that $\vars\eqsys \subseteq \dom\theta$ and, for all $\agteq\X\Gpat\in\eqsys$, $\theta(\X) = \Gpat\theta$ holds  and $\theta(\X)$ is bounded. 
  We   denote by $\gsol\eqsys$ the set of all solutions of $\eqsys$.

 A \emph{p-set equation} has shape $\agteq\x\Ppat$. We use $\eqsysP$ to range over  regular systems of p-set equations, which are defined similarly to  regular systems of type equations. 
Also  $\dom\eqsysP$, $\vars\eqsysP$ and $\gsol\eqsysP$ have the same meanings  as for systems of type equations. 

A \emph{p-condition} has shape $\pcon{\plays\X}{\x}\pp\q\participants$ and we let 
$\Pcon$ range over sets of p-conditions with pairwise distinct type and 
p-set variables. A substitution $\theta$ {\em agrees} with\\
 - $\pcon{\plays\X}{\x}\pp\q\participants$ if $\pcone{\plays{\theta(\X)}}{\theta(\x)}\pp\q\participants$;\\ 
- $\Pcon$, notation $\theta\varpropto\Pcon$, if $\theta$ agrees with all p-conditions in $\Pcon$.

We define $\gsol{\eqsys,\eqsysP,\Pcon}$ as the set of solutions of $\eqsys$ and $\eqsysP$ which agree with $\Pcon$, i.e. $\gsol{\eqsys,\eqsysP,\Pcon}=\set{\theta\in\gsol\eqsys\cap\gsol\eqsysP\mid\theta\varpropto\Pcon}$  and note that 
$\eqsys_1\subseteq\eqsys_2$, $\eqsysP_1\subseteq\eqsysP_2$, $\Pcon_1\subseteq\Pcon_2$ imply $\gsol{\eqsys_2,\eqsysP_2,\Pcon_2}\subseteq\gsol{\eqsys_1,\eqsysP_1,\Pcon_1}$.

The algorithm follows essentially the structure of coSLD resolution of coinductive logic programming \cite{Simon06,SimonBMG07,SimonMBG06,AnconaD15}, namely the extension of standard SLD resolution capable to deal with regular terms and coinductive predicates. 
A \emph{goal} is a  triple   $\pair{\Nt,\x}\X$   of a type variable $\X$, a  session  
$\Nt$ and a p-set variable $\x$. 
The algorithm takes a goal $\pair{\Nt,\x}\X$ as  input, and returns a system of type equations $\eqsys$ and a system of p-set equations $\eqsysP$ and a set of p-condition $\Pcon$. 
 A 
solution for the variable $\X$ in $\eqsys$ is a global type for the  session  
$\Nt$ in a derivation indexed by a solution for the variable $\x$ in $\eqsysP$ which satisfies the p-conditions
in $\Pcon$. 
The key idea, borrowed from coinductive logic programming, is to keep track of already encountered goals  in order to detect cycles  and so avoiding non-termination. 

The inference  judgements have  
the following shape: 
$\tyalg{\Goals}{\pair{\Nt,\x}\X}\eqsys{\eqsysP,\Pcon}$, where 
$\Goals$ is a set of goals, all with different variables  which are all  different from $\X$ and $\x$.  
Rules defining the inference algorithm are reported in Figure~\ref{fig:tyalg}.

\begin{figure}
\begin{math}
\begin{array}{c}
\NamedRule{\rn{\infn{End}}}{}{ \tyalg{\Goals}{\pair{\pP\pp\inact,\x}{\X}}{\{\agteq\X\End\}}{\{\agteq\x\emptyset\},\emptyset}}{}
\\[1.5ex]
\NamedRule{\rn{\infn{Cycle}}}{ }{ \tyalg{\Goals, \pair{\Nt,\y}\Y}{\pair{\Nt,\x}\X}{\{\agteq\X\Y\}}{\{\agteq\x\y\},\emptyset}}{}
\\[1.5ex]
\NamedRule{{\rn{\infn{Comm}}}}{
  \tyalg{ \Goals' }{\pair{\pP\pp{\PP_i}\parN\pP\q{\Q_i}\parN\Nt,\y_i}{\Y_i}}{\eqsys_i}{\eqsysP_i,\Pcon_i} \ \ \forall i \in I
}{ \tyalg{\Goals}{\pair{\pP{\pp}{\PP} \parN \pP{\q}{\Q} \parN\Nt,\x}{\X}}{\eqsys} {\eqsysP,\Pcon}}
{\begin{array}{l}
 \Goals'=\Goals, \pair{\pP{\pp}{\PP}\parN \pP{\q}{\Q} \parN \Nt, 
  x
 }{\X} \\
\PP = \Pou\q\lambda\PP i I \\
\Q = \Pin\pp\lambda\Q j J\qquad I\subseteq J\\
\Y_i,\y_i\  \mathrm{fresh}\ \forall i \in I \\
\eqsys = \{\agteq\X{\Gc\pp\q\lambda\Y i I}\}\cup\bigcup_{i \in I} \eqsys_i  \\
\eqsysP = \{\agteq\x{\bigcup_{i \in I}\y_i}\}\cup\bigcup_{i \in I} \eqsysP_i  \\
\Pcon=\set{\pcon{\plays{\Y_i}}{\y_i}\pp\q{\plays{\Nt}}\mid\forall i \in I}\\
~~~~~~~~\cup\bigcup_{i \in I} \Pcon_i 
\end{array}}
\\[1.5ex]
 \NamedRule{\rn{\infn{\sw}}}{\tyalg{\Goals,\pair{\Nt_1\parN\Nt_2,\x}{\X}}{\pair{\Nt_1,\y}{\Y}}{\eqsys_1}{\eqsysP_1,\Pcon}}{ \tyalg{\Goals}{\pair{\Nt_1\parN\Nt_2,\x}{\X}}{\eqsys}{\eqsysP,\Pcon}}{\Y,\y\  \mathrm{fresh}\\
 \participants=\plays{\Nt_2}\neq\emptyset\\
 \eqsys = \set{\agteq\X\Y}\cup \eqsys_1 \\
 \eqsysP = \set{\agteq\x{\y\cup\participants}}\cup \eqsysP_1 }
\end{array}
\end{math}
\caption{Rules of the inference algorithm.}
\label{fig:tyalg}
\end{figure}

For a terminated  session  the algorithm returns just the  two equations $\agteq\X\End$ and $\agteq\x\emptyset$ and the empty set of conditions   (Axiom \rn{\infn{End}}). 

In Rule \rn{\infn{Comm}} the algorithm nondeterministically
selects one of the matching pairs of processes:
$\PP = \Pou\q\lambda\PP i I\ \text{and}\
\Q = \Pin\pp\lambda\Q j J\text{, with}\ I\subseteq J$, i.e. two participants willing to communicate such
that the output process can freely choose the message.
The algorithm is then recursively applied, for each $i\in I$,  to the session 
where the processes of $\pp$ and $\q$ are, respectively, 
$\PP_i$ and $\Q_i$. 
In each call the goal $\pair{\pP{\pp}{\PP}\parN \pP{\q}{\Q} \parN \Nt, x }{\X}$ is added to the set of goals.
At the end of the recursive calls 
the algorithm collects all the resulting equations plus another  two for the current variables.
  Note  that variables for  the goals in the premises are fresh.
 This is important to ensure that the sets of equations $\eqsys$ and $\eqsysP$ in the conclusion are indeed regular systems of equations  (there is at most one equation for each variable). 
 The  new p-condition  ensures that
the resulting global type associated to
$\X$  and the resulting set of participants associated to
$\x$ satisfy  the conditions on participants
required by Rule \rn{Comm} in Definition~\ref{def:type-system}.

In Rule \rn{\infn{\sw}}  the algorithm nondeterministically partitions the input session into two subsessions $\Nt_1$ and $\Nt_2$ and then it is recursively called on the former. After the recursive call it simply adds the same participants to the session (together with their processes) and to the set of ignored participants  by means of the equation $\agteq\x{\y\cup\participants}$. 

 Finally,  Axiom  \rn{\infn{Cycle}} detects cycles: if the session in the current goal appears also in the set  $\Goals$, the algorithm can stop,
 returning just two equations unifying the type and p-set variables associated with the session together with the empty set of conditions. 
 
 \begin{example}[Inference for the social media] \em Figure~\ref{ifsm} gives a type inference where:\\
 - the processes 
 $\PP$, $\Q$, $\PU$, $\PP_1$, $\Q_1$, $\PU_1$, $\PP_2$, $\Q_2$ are defined as in Example~\ref{ex:typingwe} and  $\PU_2=\q!\msg{dnd}.\PU$,
 $\PU_3=\q!\msg{grtd}.\PU$;\\ 
 - the goals are
 
  \Cline{\Goals_1=\set{\triple\X{\pP\pp{\PP}\parN \pP\q{\PQ}\parN \pP\pu{\PU}}\x}\qquad 
 \Goals_2=\Goals_1\cup\set{\triple{\Y_1}{\pP\pp{\PP_1}\parN \pP\q{\PQ_1}\parN \pP\pu{\PU}}{\y_1}}}  \Cline{\Goals_3=\Goals_4=\Goals_2\cup\set{\triple{\Y_2}{\pP\pp{\pu?\Set{\msg{dnd}.\PP,\, \msg{grtd}.\PP_2}}\parN \pP\q{\PQ_1}\parN \pP\pu{\PU_1}}{\y_2}}}  \Cline{\Goals_5=\Goals_3\cup\set{\triple{\Y_3}{\pP\pp{\PP}\parN \pP\q{\PQ_1}\parN \pP\pu{\PU_2}}{\y_3}}\qquad 
 \Goals_6=\Goals_4\cup\set{\triple{\Y_4}{\pP\pp{\PP_2}\parN \pP\q{\PQ_1}\parN \pP\pu{\PU_3}}{\y_4}}}
  \Cline{\Goals_7=\Goals_6\cup\set{\triple{\Y_6}{\pP\pp{\PP_2}\parN \pP\q{\PQ_2}\parN \pP\pu{\PU}}{\y_6}}\qquad
 \Goals_8=\Goals_7\cup$ \linebreak $\set{\triple{\Y_7}{\pP\pp{\PP_2}\parN \pP\q{\PQ_2}}{\y_7}}}
 
 - the systems of type equations are  
 
  \Cline{\eqsys=\set{\agteq\X{\q\to\pp{:}\msg{hello}.\Y_1}}\cup\eqsys_1\qquad
 \eqsys_1=\set{\agteq{\Y_1}{\pp\to\pu{:}\msg{req}.\Y_2}}\cup\eqsys_2} 
  \Cline{\eqsys_2=\set{\agteq{\Y_2}{\pu\to\pp{:}\set{\msg{dnd}.\Y_3,\msg{grtd}.\Y_4}}}\cup\eqsys_3\cup\eqsys_4\qquad
 \eqsys_3=\set{\agteq{\Y_3}{\pu\to\q{:}\msg{dnd}.\Y_5}}\cup\eqsys_5}
  \Cline{\eqsys_4=$\linebreak $\set{\agteq{\Y_4}{\pu\to\q{:}\msg{grtd}.\Y_6}}\cup\eqsys_6\qquad\eqsys_5=\set{\agteq{\Y_5}\X}} 
  \Cline{\eqsys_6=\set{\agteq{\Y_6}{\Y_7}}\cup\eqsys_7\qquad
 \eqsys_7=\set{\agteq{\Y_7}{\pp\to\q{:}\msg{hello}.\Y_8}}\cup\eqsys_8\qquad\eqsys_8=\set{\Y_8=\End}}
 
  - the systems of p-set equations are  
  
   \Cline{\eqsysP=\set{\agteq\x{\y_1}}\cup\eqsysP_1\qquad\eqsysP_1=\set{\agteq{\y_1}{\y_2}}\cup\eqsysP_2\qquad
  \eqsysP_2=\set{\agteq{\y_2}{\y_3\cup\y_4}}\cup\eqsysP_3\cup\eqsysP_4} 
  \Cline{\eqsysP_3=\set{\agteq{\y_3}{\y_5}}\cup\eqsysP_5
   \qquad\eqsysP_4=\set{\agteq{\y_4}{\y_6}}\cup\eqsysP_6\qquad\eqsysP_5=\set{\agteq{\y_5}{\x}}}
   \Cline{\eqsysP_6=\set{\agteq{\y_6}{\y_7\cup\set\pu}}\cup\eqsysP_7\qquad\eqsysP_7=\set{\agteq{\y_7}{\y_8}}\cup\eqsysP_8\qquad\eqsysP_8=\set{\agteq{\y_8}{\emptyset}}}
   
  - the sets of p-conditions are 
  
   \Cline{\Pcon=\set{\agteq{(\plays{\Y_1}\cup\y_1)\setminus\set{\pp,\q}}{\set\pu}}\cup\Pcon_1\qquad
  \Pcon_1=\set{\agteq{(\plays{\Y_2}\cup\y_2)\setminus\set{\pp,\pu}}{\set\q}}\cup\Pcon_2}
   \Cline{\Pcon_2=\set{\agteq{(\plays{\Y_3}\cup\y_3)\setminus\set{\pu,\pp}}{\set\q}, \agteq{(\plays{\Y_4}\cup\y_4)\setminus\set{\pu,\pp}}{\set\q}}\cup\Pcon_3\cup\Pcon_4}
  \Cline{\Pcon_3=$\linebreak $\set{\agteq{(\plays{\Y_5}\cup\y_5)\setminus\set{\pu,\q}}{\set\pp}}\cup\Pcon_5\qquad
  \Pcon_4=\set{\agteq{(\plays{\Y_6}\cup\y_6)\setminus\set{\pu,\q}}{\set\pp}}\cup\Pcon_6\qquad\Pcon_5=\emptyset} 
   \Cline{\Pcon_6=\set{\agteq{\y_6}{\y_7\cup\set\pu}}\cup\Pcon_7\qquad
  \Pcon_7=\set{\agteq{(\plays{\Y_8}\cup\y_8)\setminus\set{\pp,\q}}{\emptyset}}\cup\Pcon_8\qquad\Pcon_8=\emptyset}
 
 \noindent
  One can easily verify that a solution of both systems of equations  $\eqsys$ and $\eqsysP$  satisfying the p-conditions   (i.e. which agrees with $\Pcon$)  is $\X=\G$ and $\x=\set\pu$, where $\G$ is the global type defined in Example~\ref{ex:gtex} and derived for this session in Figure~\ref{fsm}.\finex
 \end{example}
 
 \begin{figure}
$\prooftree
\prooftree
         \prooftree
                  \prooftree
                         \prooftree
                         \justifies
                            \tyalg{\Goals_5}{\triple{\Y_5}{\pP\pp{\PP}\parN \pP\q{\PQ}\parN \pP\pu{\PU}}{\y_5}}{\eqsys_5}{\eqsysP_5,\Pcon_5}
                            \endprooftree
                         \justifies
                              \tyalg{\Goals_3}{\triple{\Y_3}{\pP\pp{\PP}\parN \pP\q{\PQ_1}\parN \pP\pu{\PU_2}}{\y_3}}{\eqsys_3}{\eqsysP_3,\Pcon_3}
                         \endprooftree     
                         \qquad
                         \prooftree
                                 \prooftree
                                          \prooftree
                                                 \prooftree
                                                 \justifies
                                                      \tyalg{\Goals_8}{\triple{\Y_8}{\pP\pp{\inact}\parN \pP\q{\inact}}{\y_8}}{\eqsys_8}{\eqsysP_8,\Pcon_7}
                                                 \endprooftree
                                           \justifies
                                                    \tyalg{\Goals_7}{\triple{\Y_7}{\pP\pp{\PP_2}\parN \pP\q{\PQ_2}}{\y_7}}{\eqsys_7}{\eqsysP_7,\Pcon_6}
                                           \endprooftree  
                                 \justifies
                                      \tyalg{\Goals_6}{\triple{\Y_6}{\pP\pp{\PP_2}\parN \pP\q{\PQ_2}\parN \pP\pu{\PU}}{\y_6}}{\eqsys_6}{\eqsysP_6,\Pcon_6}
                                  \endprooftree   
                         \justifies
                              \tyalg{\Goals_4}{\triple{\Y_4}{\pP\pp{\PP_2}\parN \pP\q{\PQ_1}\parN \pP\pu{\PU_3}}{\y_4}}{\eqsys_4}{\eqsysP_4,\Pcon_4}
                          \endprooftree   
                  \justifies
                     \tyalg{\Goals_2}{\triple{\Y_2}{\pP\pp{\pu?\Set{\msg{dnd}.\PP,\, \msg{grtd}.\PP_2}}\parN \pP\q{\PQ_1}\parN \pP\pu{\PU_1}}{\y_2}}{\eqsys_2}{\eqsysP_2,\Pcon_2}
                  \endprooftree
         \justifies
             \tyalg{\Goals_1}{\triple{\Y_1}{\pP\pp{\PP_1}\parN \pP\q{\PQ_1}\parN \pP\pu{\PU}}{\y_1}}{\eqsys_1}{\eqsysP_1,\Pcon_1}
          \endprooftree
\justifies
          \tyalg{}{\triple\X{\pP\pp{\PP}\parN \pP\q{\PQ}\parN \pP\pu{\PU}}\x}{\eqsys}{\eqsysP,\Pcon}
\endprooftree
$
\caption{A type inference for  the social media.}\label{ifsm}
 \end{figure}

Let $\eqsys$,   $\eqsysP$ be two systems of type and p-set equations, $\Pcon$  a set of p-conditions and $\Goals$ a set of goals. 
A solution $\theta \in \gsol{\eqsys,\eqsysP,\Pcon}$   {\em agrees} with $\Goals$ if 
$\pair{\Nt,\x}\X\in\Goals$ implies $\plays{\theta(\X)}\cup\theta(\x) = \plays\Nt$ for all $\X\in\vars\eqsys$ and all $\x\in\vars\eqsysP$. 
We   denote by $\gsol[\Goals]{\eqsys,\eqsysP,\Pcon}$ the set of all solutions in $\gsol{\eqsys,\eqsysP,\Pcon}$  agreeing with $\Goals$. 
We say that a system of  type equations $\eqsys$ is {\em guarded} if 
  $\agteq\X\Y$ and $ \agteq\Y\Gpat$ in $\eqsys$ imply   that $\Gpat$ is not a variable. 
Moreover,  $\eqsys$ is \emph{$\Goals$-closed}  if it is guarded and 
$\dom\eqsys\cap\vars\Goals = \emptyset$ and $\vars\eqsys\setminus\dom\eqsys \subseteq\vars\Goals$. 
 We define similarly when a set of  p-set  equations $\eqsysP$ is guarded and $\Goals$-closed.  

Toward proving properties of the inference algorithm,   we check a couple of auxiliary lemmas.  As usual $\tyalg{\Goals}{\pair{\Nt,\x}\X}{\eqsys}{\eqsysP,\Pcon}$ means that this judgment 
 belongs to a derivation in the system of  Figure~\ref{fig:tyalg}  having a judgment with an empty sets of goals as conclusion (namely it represents the result of a recursive call during the execution of our
algorithm). 

\begin{lemma}\label{lem:vars-eq}
If $\tyalg{\Goals}{\pair{\Nt,\x}\X}{\eqsys}{\eqsysP,\Pcon}$, then 
$\eqsys$ and $\eqsysP$ are $\Goals$-closed. 
\end{lemma}
\begin{proof}
By a straightforward induction on the derivation of $\tyalg{\Goals}{\pair{\Nt,\x}\X}{\eqsys}{\eqsysP,\Pcon}$. \qed
\end{proof} 

\begin{lemma}\label{lem:eq-play}
If $\tyalg{\Goals}{\pair{\Nt,\x}\X}{\eqsys}{\eqsysP,\Pcon}$ and $\theta \in \gsol[\Goals]{\eqsys,\eqsysP,\Pcon}$, then $\plays{\theta(\X)}\cup\theta(\x)=\plays{\Nt}$.
\end{lemma}
\begin{proof}
By induction on the derivation of $\tyalg{\Goals}{\pair{\Nt,\x}\X}{\eqsys}{\eqsysP,\Pcon}$. The only interesting case is when Rule \rn{\infn{Comm}} is applied. 
From $\theta \varpropto\Pcon$ we get $(\plays{\theta
(\Y_i)}\cup\theta(\y_i))\setminus\set{\pp,\q}={\plays{\Nt}}$ for all $i \in I$, which imply $\plays{\theta
(\X)}\cup\theta(\x)={\plays{\pP{\pp}{\PP}\parN \pP{\q}{\Q} \parN \Nt}}$ since $\agteq\X{\Gc\pp\q\lambda\Y i I}\in  \eqsys$ and 
$\agteq\x{\bigcup_{i \in I}\y_i}\in\eqsysP$.  \qed
\end{proof}

\begin{figure}
 \begin{math}
 \begin{array}{c}
 \NamedRule{\rn{\itrp{End}}}{}{ \tynIP\Nset{\pP\pp\inact:\End}\emptyset }{} \qquad 
 \NamedRule{\rn{\itrp{Cycle}}}{}{ \tynIP{\Nset,\ipair\G\Nt\participants}{\Nt:\G} \participants}{} \\[1.5ex]
 \NamedRule{\rn{\itrp{Comm}}}{\mbox{$\begin{array}{c}
   \tynIP{\Nset,\ipair\G\Nt\participants}{\pP\pp{\PP_i}\parN\pP\q{\Q_i}\parN\Nt':\G_i}{\participants_i} \\ (\plays{\G_i}\cup\participants_i)\setminus\set{\pp,\q}=\plays{\Nt'}
   \ \ \forall i \in I\end{array}$}
 }{ \tynIP{\Nset}{\Nt:\G} \participants}{\begin{array}{c}\Nt\equiv \pP\pp{\Pou\q\lambda\PP i I}\parN\pP\q{\Pin\pp\lambda\Q j J}\parN\Nt'\\
 \G=\Gc{\pp}{\q}{\lambda}{\G}i I \quad \G \mathrm{\ is\ bounded}\\ 
 \participants=\cup_{i\in I}\participants _i \quad I\subseteq J\end{array}}\\[1.5ex]
 \NamedRule{\rn{\itrp{\sw}}}{\tynIP{\Nset,\ipair\G{\Nt_1\parN\Nt_2
 }{\participants_1\cup\participants_2
 }}{\Nt_1:\G}{\participants_1}}{ \tynIP\Nset{\Nt_1\parN\Nt_2:\G}{\participants_1\cup\participants_2} }{\participants_2=\plays{\Nt_2}\neq\emptyset}
 \end{array}
 \end{math}
 \caption{Inductive typing rules for sessions.}\label{fig:itr}
 \end{figure}
\medskip
To show soundness   and completeness   of our inference algorithm, it is handy to formulate an inductive version of our typing rules, see Figure~\ref{fig:itr},
where $\Nset$ ranges over sets of triples $\ipair\G\Nt\participants$.
We can give an inductive formulation since all infinite derivations using the typing rules of  Definition~\ref{def:type-system}  are regular, i.e. the number of different subtrees of a derivation for a judgement $\tyng\G\Nt\participants$ is finite.  In fact,   it is bounded by the product of the number of different subterms of $\G$ and the number of different subsessions 
of $\Nt$, which are both finite as $\G$ and (processes in) $\Nt$ are regular.
 Applying  the standard transformation according to~\cite[Section  21.9]{pier02} from a coinductive to an inductive formulation we get the typing rules shown in Figure~\ref{fig:itr}.  
 
  In the following two lemmas we relate inference and inductive  derivability.  

 \begin{lemma}\label{lem:inf-sound}
If $\tyalg{\Goals}{\pair{\Nt,\x}\X}{\eqsys}{\eqsysP,\Pcon}$,  then    $\tynI{\Goals\theta}{\Nt : \theta(\X)}{ \theta(x) }$ 
for all $\theta \in \gsol[\Goals]{\eqsys,\eqsysP,\Pcon}$ such that $\vars\Goals\subseteq\dom\theta$.  \end{lemma} 
\begin{proof}
By induction on the derivation of $\tyalg{\Goals}{\pair{\Nt,\x}\X}{\eqsys}{\eqsysP,\Pcon}$. 
 
{\em Axiom} \rn{\infn{End}} 
We have $\eqsys = \set{\agteq\X\End}$, $\eqsysP = \set{\agteq\x\emptyset}$ and $\Pcon = \emptyset$, hence $\theta(\X) = \End$,  $\theta(\x) =  \emptyset$ and the thesis follows by  Axiom $\rn{\itr{End}}$. 

{\em Axiom} \rn{\infn{Cycle}} 
We have $\eqsys = \set{\agteq\X\Y}$, $\eqsysP = \set{\agteq\x\y}$, $\Pcon = \emptyset$, and $\Goals = \Goals',\pair{\Nt,\y}\Y$. 
Then, $\theta(\X) = \theta(\Y)$, $\theta(\x) = \theta(\y)$ and   the thesis follows by  Axiom  \rn{\itr{Cycle}}.  

{\em Rule} \rn{\infn{Comm}}
We have $\Nt \equiv \pP\pp{\Pou\q\lambda\PP i I}\parN\pP\q{\Pin\pp\lambda\Q j J} \parN \Nt'$ with $I\subseteq J$  and   $\tyalg{\Goals,\pair{\Nt,\x}\X}{\pair{\Nt_i,\y_i}{\Y_i}}{\eqsys_i}{\eqsysP_i,\Pcon_i}$ with $\Y_i,\y_i$ fresh and $\Nt_i \equiv \pP\pp{\PP_i} \parN\pP\q{\Q_i} \parN\Nt'$   and 
$\eqsys = \{\agteq\X{\Gc\pp\q\lambda\Y i I}\}\cup\bigcup_{i \in I} \eqsys_i $, 
$\eqsysP = \{\agteq\x{\bigcup_{i \in I}\y_i}\}\cup\bigcup_{i \in I} \eqsysP_i $, 
$\Pcon=\set{\pcon{\plays{\Y_i}}{\y_i}\pp\q{\plays{\Nt'}}\mid\forall i \in I}\cup\bigcup_{i \in I} \Pcon_i$. 
Since $\eqsys_i\subseteq\eqsys$, $\eqsysP_i\subseteq\eqsysP$, and $\Pcon_i\subseteq\Pcon$, we have $\theta \in \gsol{\eqsys_i,\eqsysP_i,\Pcon_i}$ for all $i\in I$.  Being $\theta \in \gsol[\Goals]{\eqsys,\eqsysP,\Pcon}$, Lemma~\ref{lem:eq-play} implies $\plays{\theta(\X)}\cup\theta(\x) = \plays{\Nt}$. So   we get that $\theta$ agrees with $\Goals,\pair{\Nt,\x}\X$.
Then, by the induction hypothesis, we   have
$\tynI{\Goals\theta, \ipair {\theta(\X)}\Nt{\theta(\x)}}{\Nt_i : \theta(\Y_i)}{\theta(\y_i)}$ for all $i \in I$. The thesis follows by Rule \rn{\itr{Comm}}, since
$\theta(\X) = \pp\to\q:\set{\lambda_i.\theta(\Y_i)}_{i\in I}$ and $\theta( \x) = \bigcup_{i \in I}\theta( \y_i )$.

{\em Rule} \rn{\infn{\sw}} We have 
$\Nt\equiv\Nt_1\parN\Nt_2$ and $\participants=\plays{\Nt_2}\neq\emptyset$ and  
$\tyalg{\Goals,\pair{\Nt_1\parN\Nt_2,\x}{\X}}{\pair{\Nt_1,\y}{\Y}}{\eqsys_1}{\eqsysP_1,\Pcon}$ and $\eqsys = \set{\agteq\X{\Y}}\cup \eqsys_1$  and $\eqsysP = \set{\agteq\x{\y\cup\participants}}\cup \eqsysP_1$. 
Being $\theta \in \gsol[\Goals]{\eqsys,\eqsysP,\Pcon}$, Lemma~\ref{lem:eq-play} implies $\plays{\theta(\X)}\cup\theta(\x) = \plays{\Nt}$. So   we get that $\theta$ agrees with $\Goals,\pair{\Nt,\x}\X$.
Then, by the induction hypothesis, we  have  $\tynI{\Goals\theta, \ipair {\theta(\X)}\Nt{\theta(\x)}}{\Nt_1 : \theta(\Y)}{\theta(\y)}$. 
The thesis follows by Rule \rn{\itr{\sw}}. \qed
\end{proof}

\begin{lemma}\label{lem:inf-complete}
If $\tynI{\Nset}{\Nt:\G}\participants$ and   $\plays{\G'}\cup\participants'=\plays{\Nt'}$ for all $\pair{\Nt',\participants'}{\G'}\in\Nset$,   then, for all $\Goals$, $\X$, $\x$ and $\sigma$  such that $\X,\x\notin\vars\Goals$, $\dom\sigma = \vars\Goals$ and $\Goals\sigma = \Nset$,  there are $\eqsys$, $\eqsysP$, $\Pcon$ and $\theta$ such that
$\tyalg{\Goals}{\pair{\Nt,\x}\X}{\eqsys}{\eqsysP,\Pcon}$ and $\theta\in\gsol[\Goals]{\eqsys,\eqsysP,\Pcon}$ and
$\dom\theta = \vars\eqsys\cup \vars\eqsysP\cup\vars\Goals$ and  
$\sigma\sbtord\theta$ and $\theta(\X) = \G$  and $\theta(\x) = \participants$. 
\end{lemma}
\begin{proof}
By induction on the derivation of $\tynI{\Nset}{\Nt:\G}\participants$. It is easy to verify that $\tynI{\Nset}{\Nt:\G}\participants$ implies $\plays\G\cup\participants=\plays\Nt$.
%

{\em Axiom} \rn{\itr{End}} 
The thesis is immediate by  Axiom  \rn{\infn{End}} taking $\theta = \sigma \sbtplus \set{\X\mapsto\End,\x\mapsto\emptyset}$. 

 {\em Axiom} \rn{\itr{Cycle}}  In this case 
we  have   $\Nset = \Nset',\ipair  \G\Nt\participants $,   then $\Goals = \Goals', \pair{\Nt,\y}\Y$ and $\sigma(\Y) = \G$ and $\sigma(\y) = \participants$. 
By  Axiom  \rn{\infn{Cycle}}, we get $\tyalg{\Goals}{\pair{\Nt,\x}\X}{\set{\agteq\X\Y}}{\set{\agteq\x\y},\emptyset}$, hence 
$\theta = \sigma \sbtplus \set{\X\mapsto \G, \x\mapsto \participants}$ is a solution of $\set{\agteq\X\Y}$ and of $\set{\agteq\x\y}$, which agrees with $\Goals$,   being $\plays\G\cup\participants=\plays\Nt$   as needed. 

{\em Rule} \rn{\itr{Comm}}
  In this case   we have $\Nt \equiv \pP\pp{\Pou\q\lambda\PP i I}\parN\pP\q{\Pin\pp\lambda\Q j J} \parN \Nt'$ with $I\subseteq J$ and $\G = \Gc\pp\q\lambda\G{i}{I}$ and
$\tynI{\Nset,\ipair  \G\Nt \participants}{\Nt_i:\G_i}{\participants_i}$ with $\Nt_i \equiv \pP\pp{\PP_i}\parN\pP\q{\Q_i}\parN\Nt'$   and \linebreak $(\plays{\G_i}\cup\participants_i)\setminus\set{\pp,\q}= \plays{\Nt'}$,   for all $i \in I$.  This last condition implies $\plays\G\cup\participants=\plays\Nt$.
Set $\sigma' = \sigma \sbtplus \set{\X\mapsto \G,\x\mapsto \participants}$  and $\Goals'=\Goals, \pair{\Nt,\x}\X$,  then, by the induction hypothesis, we get
that  there are $\eqsys_i, \eqsysP_i, \Pcon_i$ and  $\theta_i$ such that
$\tyalg{\Goals'}{\pair{\Nt_i,\y_i}{\Y_i}}{\eqsys_i}{\eqsysP_i, \Pcon_i}$ and $\theta_i \in \gsol[\Goals']{\eqsys_i,\eqsysP_i, \Pcon_i}$ and
 $\dom{\theta_i}= \vars{\eqsys_i}\cup \vars{\eqsysP_i}\cup\vars{\Goals'} $ and
$\sigma'\sbtord \theta_i$ and $\theta_i(\Y_i) = \G_i$ and $\theta_i(\y_i) = \participants_i$, for all $i \in I$.
We can assume that   $j\ne l$ implies $\Y_j\ne\Y_l$ and $\dom{\eqsys_j}\cap\dom{\eqsys_l} = \emptyset$ and $\y_j\ne\y_l$ and $\dom{\eqsysP_j}\cap\dom{\eqsysP_l} = \emptyset$ for all $j,l\in I$,   because the algorithm always introduces fresh variables.
  This   implies $\dom{\theta_j}\cap\dom{\theta_l} = \set{\X,\x}$ 
  for all $j \ne l$, and so $\theta = \sum_{i \in I} \theta_i$ is well defined.
Moreover, we have $\theta \in \gsol[\Goals']{\eqsys_i,\eqsysP_i, \Pcon_i}$ and
$\sigma \sbtord \theta$ and $\theta(\X) = \G$  and $\theta(\x) = \participants$, as $\sigma\sbtord \sigma'$ and $\sigma'\sbtord \theta_i \sbtord \theta$ for all $i \in I$.
  From $(\plays{\G_i}\cup\participants_i)\setminus\set{\pp,\q}= \plays{\Nt'}$ we get $(\plays{\theta(\Y_i)}\cup\theta(\y_i))\setminus\set{\pp,\q} = \plays{\Nt'}$
 for all $i \in I$.
By Rule \rn{\infn{Comm}} we get $\tyalg{\Goals}{\pair{\Nt,\x}\X}{\eqsys}{\eqsysP,\Pcon}$ with
$\eqsys = \set{\agteq\X{\Gc\pp\q\lambda\Y{i}{I}}}\cup\bigcup_{i \in I} \eqsys_i$ and 
$\eqsysP = \set{\agteq\x{\cup_{i\in I}\y_i}}\cup\bigcup_{i \in I} \eqsysP_i$ and 
$\Pcon=\set{\pcon{\plays{\Y_i}}{\y_i}\pp\q{\plays{\Nt}}\mid\forall i \in I}\cup\bigcup_{i \in I} \Pcon_i$  and
$\theta \in \gsol[\Goals]{\eqsys,\eqsysP,\Pcon}$, since
$\theta(\X) =\G= \Gc\pp\q\lambda\G{i}{I} = \pp\to\q:\set{\lambda_i.\theta_i(\Y_i)}_{i \in I} = (\Gc\pp\q\lambda\Y{i}{I})\theta$ 
and $\theta(\x) = \participants=\cup_{i\in I}\participants_i = \cup_{i\in I}\theta(\y_i)= (\cup_{i\in I}\y_i)\theta
$ and 
$\sigma\sbtord\theta$.

{\em Rule} \rn{\itr{\sw}} 
We have  ${\tynIP{\Nset,\ipair\G{\Nt_1\parN\Nt_2}{\participants_1\cup\participants_2}}{\Nt_1:\G}{\participants_1}}$
 and $\participants_2=\plays{\Nt_2}\neq\emptyset$ and $\plays\G\cup\participants=\plays\Nt$,
 where
$\Nt\equiv\Nt_1\parN\Nt_2$ and $\participants=\participants_1\cup\participants_2$.
Set $\sigma' = \sigma \sbtplus \set{\X\mapsto \G, \x\mapsto\participants_1\cup\participants_2}$  and $\Goals'=\Goals, \pair{\Nt,\x}\X$,  then, by the induction hypothesis, we get
that  there are $\eqsys_1$, $\eqsysP_1$, $\Pcon_1$ and  $\theta$  such that
$\tyalg{\Goals'}{\pair{\Nt_1,\y}{\Y}}{\eqsys_1}{\eqsysP_1,\Pcon_1}$ and $\theta \in \gsol[\Goals']{\eqsys_1,\eqsysP_1,\Pcon_1}$ and
 $\dom{\theta}= \vars{\eqsys_1}\cup \vars{\eqsysP_1}\cup \vars{\Goals'} $ and
$\sigma'\sbtord \theta$ and $\theta(\Y) = \G$  and $\theta(\y) = \participants_1$.
By Rule \rn{\infn{\sw}} we get $\tyalg{\Goals}{\pair{\Nt,\x}{\X}}{\eqsys}{\eqsysP,\Pcon_1}$ with $\eqsys=\set{\agteq\X{\Y}}\cup\eqsys_1$ and $\eqsysP=\set{\agteq\x{\y\cup\participants_2}}\cup\eqsysP_1$ and
$\theta \in \gsol[\Goals]{\eqsys,\eqsysP,\Pcon_1}$,  since
$\theta(\X) =\G= \theta(\Y) $
 and $\theta(\x) = \participants=\participants_1\cup\participants_2 =\theta(\y)\cup \participants_2= (\y\cup \participants_2)\theta$ and
$\sigma\sbtord\theta$.  \qed
\end{proof}

\begin{theorem}[Soundness and  Completeness   of Inference]\label{thm:sac}~\begin{enumerate}
\item\label{thm:sac1} If $\tyalg{}{\pair{\Nt,\x}\X}\eqsys{\eqsysP,\Pcon}$, then $\tyng{\theta(\X)}\Nt{\theta(\x)}$ for all $\theta \in\gsol{\eqsys,\eqsysP,\Pcon}$.
 \item\label{thm:sac2} If $\tyng\G\Nt\participants$, then there are $\eqsys$, $\eqsysP$, $\Pcon$ and $\theta$ such that $\tyalg{}{\pair{\Nt,\x}{\X}}{\eqsys}{\eqsysP,\Pcon}$ and $\theta \in\gsol{\eqsys,\eqsysP,\Pcon}$ and $\theta(\X)=\G$ and $\theta(\x)=\participants$. 
\end{enumerate}
\end{theorem}
\begin{proof}
(\ref{thm:sac1}). By Lemma~\ref{lem:inf-sound} $\tyalg{}{\pair{\Nt,\x}\X}{\eqsys}{\eqsysP,\Pcon}$ implies $\tynI{}{\Nt:\theta(\X)}{\theta(\x)}$ for all $\theta \in\gsol{\eqsys,\eqsysP,\Pcon}$. This is enough, since $\tynI{}{\Nt:\theta(\X)}{\theta(\x)}$ gives $\tyng{\theta(\X)}\Nt{\theta(\x)}$.

(\ref{thm:sac2}). From $\tyng\G\Nt\participants$ we get $\tynI{}{\Nt:\G}\participants$. 
By Lemma~\ref{lem:inf-complete} this implies that there are $\eqsys$, $\eqsysP$, $\Pcon$ and $\theta$ such that 
$\tyalg{}{\pair{\Nt,\x}\X}{\eqsys}{\eqsysP,\Pcon}$ and $\theta \in\gsol{\eqsys,\eqsysP,\Pcon}$ and $\theta(\X)=\G$ and $\theta(\x)=\participants$.  \qed
\end{proof}

\begin{remark}[Termination]\em
As  it  happens for (co)SLD-resolution in logic programming, the termination of the inference algorithm depends on the choice of a resolution strategy. 
Indeed, we have many sources of non-determinism: 
we have to select 
 two participants  of the session with matching processes and expand them using Rule  \rn{\infn{Comm}}, or ignore part of the session using Rule  \rn{\infn{\sw}}
or try to close a cycle using the  Axiom  \rn{\infn{Cycle}}. 
A standard way to obtain a sound and complete resolution strategy is to build a tree where all such choices are performed in  parallel and then visit the tree using a breadth-first strategy. 
The tree is potentially infinite in depth, but it is finitely branching, since at each point we have only finitely many different choices, hence this strategy necessarily enumerates all solutions. 
\end{remark}

\begin{remark}[Use of Rule \rn{\infn{\sw}}]\em
 Note  that in a $\,\tynIP{}{}{\participants}$ derivation the triple in the premise of Rule  \rn{\itrp{\sw}} can never be used in an application of  Axiom  \rn{\itrp{Cycle}}. This 
-- as already hinted at in Example~\ref{ex:typingwe} --  
immediately implies that Rule \rn{\sw} is not strictly necessary inside infinite branches of $\,\tyng{}{}{\participants}$ derivations.
Moreover, a slight simplification of the algorithm can be got since, in the step corresponding to Rule
\rn{\infn{\sw}}, we could avoid adding the goal $\pair{\Nt_1\parN\Nt_2,\x}{\X}$ to the current
set of goals. This would reduce the number of goals to be checked during the step corresponding 
to  Axiom \rn{\itrp{Cycle}}. Rule  \rn{\infn{\sw}}  turns out to  
be necessary,  instead, when applying the algorithm to sessions where non-ignored participants 
expose a finite behaviour, like  $\pP\pp{\PP_2}\parN \pP\q{\PQ_2}\parN \pP\pu{\PU}$ in Example~\ref{ex:typingwe}.
 Also the typing of 
stuck sessions  with recursive processes 
 like 
$\pP\pp\PP\parN\pP\q\Q$ where $\PP=\q!\la.\PP$ and $\Q=\pp!\la.\Q$ requires the use of Rule \rn{\infn{\sw}}. 
%
\finex
\end{remark}

\section{Concluding Remarks, Related and Future Works}\label{rfw}

Lock-freedom is definitely a relevant and widely investigated 
communication property of concurrent and distributed systems. 
It  ensures absence of \emph{locks}, a lock being a reachable configuration 
where a communication action of a participant remains pending in any possible
continuation of the system. In case the participant prevented to progress be $\pp$, 
such configuration is called a  $\pp$-\emph{lock} (see \cite{BLT22} for an abstract definition of Lock-freedom).
Lock-freedom corresponds to the notion of liveness
in~\cite{KobayashiS10,LangeNTY17}  where the synchronous communication is channel-based. 
Sometimes properties different from what we intend are referred to by ``Lock-freedom'':
for instance the notion of Lock-freedom in~\cite{Kobayashi02}, under fair scheduling,
corresponds to what \cite{DS} and  \cite{BLT22} refer to as \emph{strong Lock-freedom}.

Various formalisms and methodologies  have  been developed in order to prove  \Lc-freedom 
while others do ensure \Lc-freedom by construction.
Among the former there are type assignment systems where typability entails  \Lc-freedom, 
both for asynchronous \cite{Padovani14} and synchronous \cite{BDL22} communications.

Lock-freedom is quite a strong property: it entails Deadlock-freedom, whereas the
vice versa does not hold. In several actual scenarios, lighter forms of  \Lc-freedom  would however suffice.
For instance in clients/servers scenarios where one can accept some servers to get locked after
their interactions with the clients have been completed.

In the present paper we developed a type assignment system where typability
ensures   \emph{$\participants$-excluded  \Lc-freedom}:  the absence of $\pp$-locks
for each participant $\pp$ not belonging to $\participants$.
This is achieved by means of ``partial'' typability, i.e. by disregarding typability 
of (sub)processes of participants that we can safely assume to get possibly locked.
Multiparty sessions (parallel compositions of named processes) are (partially) typed by
\emph{global types}, which in turn describe the overall interactions  inside the  multiparty sessions.
Our partial typability ensures also that the behaviours of the  non-ignored  participants 
adhere to what the global type describes. 
As far as we know, there are not other formalisms dealing with properties like 
$\participants$-excluded  \Lc-freedom.

Our partial typing is reminiscent of connecting communications, a notion introduced in~\cite{HY17} and  further investigated 
in~\cite{CDG18,CDGH19} in order to describe protocols with optional participants. The intuition behind connecting communications is that in some parts of the protocol, delimited by a choice construct, some participants
may be optional, namely they are ``invited'' to join the interaction only in some branches of the choice, by means of connecting communications. As argued in~\cite{HY17,CDG18,CDGH19}, this feature allows for a more natural description of typical communication protocols. In~\cite{CDGH19}, connecting communications also enable to express conditional delegation: this will be obtained by writing a choice where the delegation appears only in some branches of the choice, following a connecting communication. The participants offering connecting communications should be ignored in the present type system.  
An advantage of connecting communications over partial typing is that only participants offering connecting inputs can be stuck. The disadvantage is that the typing rules are more requiring, so many interesting sessions can be partially typed but cannot be typed by means of connecting communications. 

 In designing type inference we took inspiration from~\cite{DGD22}, where inputs and outputs are split in global types in order to better describe asynchronous communication.  Our inference algorithm is related as goal, but very different as methodology, to the algorithm in~\cite{LTY15}, which builds global graphs from sets of communicating finite state machines satisfying suitable conditions.  We are planning to implement our type inference algorithm. 

Unlike  many MPST formalisms in the literature, like ~\cite{Honda2016}, we  type sessions with global types without recurring to local types and projections.  It would be interesting to investigate the possibility
of extending the standard projection operator to a relation between global types
and possibly non lock-free local behaviours. Other simplifications of our calculus are the absence of values in messages and the unicity of channels. While we can easily enrich messages with values, allowing more than one channel requires sophisticated type systems in order to get \Lc-freedom~\cite{Padovani14}. 

The following example shows a further direction for investigation of partial typing, namely
to describe and analyse privacy matters.

\newtheorem{example*}{Example}
\noindent
{\bf Example}\,[{\bf Partial typing for privacy}] 
The communications written in global types can be viewed as public, while the others can be viewed as private. For example Alice and Bob want to 
discuss privately 
 which version of a game would be the most suitable for their son Carl,
who asked for it as birthday present.
Taking participants $\pa$, $\pb$ and $\pc$ to incarnate, respectively, Alice, Bob  and Carl this scenario can be represented by the session 

\Cline{\Nt\equiv\pP\pa{\pc? \msg{present}.\PP}\parN\pP\pb{\pc? \msg{present}.\Q}\parN\pP\pc{\pa! \msg{present}.\pb!\msg{present}}}

\noindent 
where $\PP=\pb!\set{\msg{bla}.\pb?\msg{bla'}.\PP,\msg{ok}}$ and $\Q=\pa?\set{\msg{bla}.\pa!\msg{bla'}.\PP,\msg{ok}}$.\\
 A suitable global type is $\G=\pc\to\pa{:}\msg{present}.\pc\to\pb{:}\msg{present}$. We can in fact derive $\tyng\G\Nt{{ \Set{\pa,\pb}}}$. 
\finex

\smallskip

 We also plan to investigate partial typing for asynchronous communications,  
possibly  modifying the type system  of~\cite{DGD22}. 
An advantage of that type system is the possibility of anticipating outputs over inputs without requiring the asynchronous subtyping of~\cite{MYH09}, which is  known to be undecidable~\cite{BCZ17,LY17}.  A difficulty will come from the larger  
freedom in choosing the order of interactions due to the splitting between writing and reading messages on a queue. 

\smallskip

{\bf Acknowledgments} We wish to gratefully thank the  anonymous reviewers for their thoughtful and helpful comments.

\bibliographystyle{eptcs}
\bibliography{session}







\end{document}